\documentclass[12pt]{article}

\usepackage{xr-hyper}
\usepackage{amsmath, amsfonts, amssymb, amsthm, fullpage, color, bbm, enumerate, enumitem, titlesec, graphicx, epstopdf, multirow, array, mathtools}
\usepackage{subcaption}
\usepackage[T1]{fontenc}
\usepackage{lmodern}
\usepackage{pdflscape}
\usepackage{booktabs}
\usepackage{longtable}
\usepackage[title]{appendix}

\usepackage[onehalfspacing]{setspace}

\definecolor{darkred}{rgb}{0.5,0,0}

\titleformat*{\paragraph}{\sc}

\usepackage{hyperref}
\hypersetup{allbordercolors={1 1 1},allcolors=darkred,colorlinks=true}
\usepackage[capitalize,nameinlink,noabbrev]{cleveref}

\usepackage[font={small}]{caption}

\usepackage[round]{natbib}
\usepackage[runin]{abstract}

\newcolumntype{C}[1]{>{\centering\arraybackslash}p{#1}}

\theoremstyle{plain}
\newtheorem{asn}{Assumption}
\crefname{asn}{Assumption}{Assumptions}
\newtheorem{lem}{Lemma}
\crefname{lem}{Lemma}{Lemmas}
\newtheorem{thm}{Theorem}
\crefname{thm}{Theorem}{Theorems}
\newtheorem{prop}{Proposition}
\crefname{prop}{Proposition}{Propositions}
\newtheorem{cor}{Corollary}
\crefname{cor}{Corollary}{Corollaries}

\crefname{conj}{Conjecture}{Conjectures}
\newtheorem{defn}{Definition}
\crefname{defn}{Definition}{Definitions}
\newtheorem*{claim*}{Claim}
\newtheorem*{cor*}{Corollary}

\theoremstyle{definition}
\newtheorem{exm}{Example}
\crefname{exm}{Example}{Examples}
\newtheorem{remark}{Remark}
\crefname{remark}{Remark}{Remarks}

\newlist{asnitem}{enumerate}{1}
\setlist[asnitem,1]{label=(\roman*), ref=\theasn(\roman*), leftmargin=2em}
\crefname{asnitemi}{Assumption}{Assumptions}
\crefalias{asnitemi}{asn}
\makeatletter
\let\cref@old@resetby\cref@resetby
\def\cref@resetby#1#2{%
\let#2\relax%
\cref@ifstreq{#1}{asnitemi}{\def#2{asn}}{}%
\ifx#2\relax%
\cref@old@resetby{#1}{#2}%
\fi
}
\DeclareRobustCommand\crefnosort[1]{\begingroup\@cref@sortfalse\cref{#1}\endgroup}
\makeatother

\AtBeginDocument{}

\DeclareMathOperator*{\argmin}{argmin}

\DeclareMathOperator*{\ve}{vec}

\DeclareMathOperator*{\var}{Var}
\DeclareMathOperator*{\cov}{Cov}

\newcommand{\E}{E}
\DeclareMathOperator{\cum}{cum}
\newcommand{\dto}{\overset{d}{\to}}
\newcommand{\pto}{\overset{p}{\to}}
\newcommand{\1}{\mathbbm{1}}

\allowdisplaybreaks

\crefname{sappsec}{Supplemental Appendix}{Supplemental Appendices}
\crefname{sappsubsec}{Supplemental Appendix}{Supplemental Appendices}
\crefname{sappsubsubsec}{Supplemental Appendix}{Supplemental Appendices}
\crefname{appsec}{Appendix}{Appendices}

\defcitealias{AuclertBardoczyRognlieStraub2021}{ABRS}

\begin{document}

\title{Limited-Information Estimation \texorpdfstring{\\}{}of Heterogeneous Agent Models\thanks{Email: {\tt laura.liu@pitt.edu}, {\tt mikkelpm@uchicago.edu}, {\tt nt1847@princeton.edu}. We are grateful for comments from St\'{e}phane Bonhomme, Mikhail Golosov, Greg Kaplan, Christian Wolf, and seminar participants at Chicago, FRB Richmond, and Princeton. Plagborg-M{\o}ller acknowledges that this material is based upon work supported by the NSF under Grant {\#}2238049 and by the Alfred P.\ Sloan Foundation.}}
\author{\begin{tabular}{ccccc}
Laura Liu && Mikkel Plagborg-M{\o}ller && Nelson Matthew P. Tan \\
{\small University of Pittsburgh} && {\small University of Chicago} && {\small Princeton University}
\end{tabular}}
\date{\texorpdfstring{\bigskip}{ }August 13, 2026}
\maketitle

\begin{abstract}
We develop a method for estimating and testing a single block of a macroeconomic model with heterogeneous agents, without placing assumptions on the structure of the rest of the economy. In a large class of models, individual agents' decisions depend on the macroeconomy only through their expectations of the evolution of a finite-dimensional vector of ``sufficient statistics'' (e.g., asset returns or aggregate earnings). Our estimator selects the structural parameters that provide the best model-consistent fit between empirical impulse responses with respect to identified macro shocks of (a) cross-sectional moments of agent choices (e.g., moments of consumption) and (b) the vector of sufficient statistics. In a simulation illustration, we estimate a two-asset heterogeneous household model block without restricting production, firm investment, financial intermediation, monetary policy, trade, etc.
\end{abstract}

\clearpage

\section{Introduction}
\label{sec:intro}

In recent years, there has been extensive research on the consequences of household and firm heterogeneity for macroeconomic dynamics \citep{KaplanMollViolante2018,AuclertRognlieStraub2025}. Models with heterogeneous agents (HAs) promise not only to fit the micro data better but also to shed new light on transmission channels and policy trade-offs. Estimation of these models has proven to be challenging due to their computational complexity and the statistical difficulties involved in jointly analyzing micro and macro data. For this reason, much of the applied literature selects structural parameters by calibrating to a set of moments, typically chosen heuristically. Nevertheless, the last decade has seen continued progress in this area, and there now exist rigorous inference methods that efficiently exploit time series of macro variables and cross-sectional or panel moments, or indeed the full information content of repeated cross sections.\footnote{Examples of the former include \citet{MongeyWilliams2017,Winberry2018,PappReiter2020,AuclertBardoczyRognlieStraub2021,Acharya_et_al2023,BayerBornLuetticke2024}. Examples of the latter include \citet{FernandezVillaverdeHurtadoNuno2023,LiuPlagborgMoller2023}.}

Whether matching an \emph{ad hoc} set of moments or using full-information likelihood procedures, a familiar issue in structural estimation is that misspecification of one part of the model can contaminate the estimation of parameters belonging to other parts of the model. For example, even if a researcher is primarily interested in the consumption block of a Heterogeneous Agent New Keynesian (HANK) model, they still have to specify the full structure of the rest of the economy in order to compute arbitrary moments or the likelihood function implied by general equilibrium.

In representative agent models, there exists a well-known robust alternative to full-information inference: Generalized Method of Moments (GMM) estimation of the optimality conditions associated with a single model block, such as the consumption Euler equation \citep{HansenSingleton1982,Yogo2004} or the New Keynesian Phillips Curve \citep{GaliGertler1999}. However, it is difficult to derive empirically usable, model-consistent moment conditions in typical HA macro models, since agents' decision rules are subject to occasionally binding constraints and depend on a mix of micro and macro variables, some of which are latent.

In this paper we develop a general method for estimating and testing a single block of a HA model in isolation, without restricting other parts of the model. We exploit a well-known feature of conventional HA macro models: individual agents' decisions depend on the macroeconomy only through a finite-dimensional vector of aggregate ``sufficient statistics'', typically prices such as asset returns or wages.\footnote{This fact has also been instrumental to the work of \citet{AuclertBardoczyRognlieStraub2021,McKayWolf2023,Pallotti_et_al2024}.} The ``sequence-space Jacobians'' (SSJs), introduced by \citet{BoppartKrusellMitman2018} and \citet[henceforth abbreviated ``ABRS'']{AuclertBardoczyRognlieStraub2021}, are the theoretical objects that measure how aggregates of individual decisions depend on current and future expected values of the macro sufficient statistics. These SSJs are functions of the underlying structural parameters of the model block under estimation, but they do not depend on other unmodeled features of the general equilibrium. Our estimator selects the structural parameters of the model block of interest so as to provide the best fit between, on the one hand, the empirical impulse responses of aggregated individual outcomes (such as cross-sectional moments of consumption) and, on the other hand, the combination of theoretical SSJs and empirical impulse responses of the sufficient statistics. Our estimator can therefore be viewed as a nonlinear version of ``regression in impulse response space'' \citep{BarnichonMesters2020}, applicable to general HA models. Given the estimated parameters, we can also carry out an over-identification test of the validity of the assumptions in the single model block.

The conceptual novelty of our approach is that we do not impose a specific process for the evolution of the macro sufficient statistics. Existing full-information estimation procedures derive the functional form of this process from the structural assumptions placed on the rest of the economy, such as the firm, finance, international trade, and government blocks. If some of these extraneous assumptions are wrong, all elements of the full-information estimator could be biased. Adopting a reduced-form model for the dynamics of the sufficient statistics does not easily solve this problem, since such a model could theoretically depend on lags of the entire infinite-dimensional distribution of idiosyncratic state variables, in addition to (latent) macro state variables. Our procedure sidesteps the issue by merely estimating the impulse responses of the sufficient statistics with respect to a set of observed economic shocks, without requiring a full specification of the transmission of every possible shock. Focusing on these particular moments allows our procedure to be fully consistent with general equilibrium, without having to impose a full equilibrium model (in the spirit of \citealp{HansenSingleton1982}).

Our ability to relax the modeling assumptions required by full-information estimation procedures relies on having access to time series data on the entire vector of macro sufficient statistics for the model block of interest, as well as a set of identified economic shocks to serve as instruments (with respect to which the impulse responses are computed). It is straightforward to obtain data on the prices and other aggregates in the vector of sufficient statistics for many existing HA models. As for instruments, researchers can exploit the menu of existing identified shocks in the applied literature \citep{Ramey2016,BarnichonMesters2020}. The instruments must induce sufficient dynamic variation in the macro sufficient statistics to ensure identification of the structural parameters. Moreover, they must be actual ``shocks'' in the sense of being orthogonal to any lagged macro variables, and in particular to the pre-determined distribution of idiosyncratic state variables. Conveniently, we do \emph{not} need the identified shocks to have a unique economic interpretation, since any functions of multiple contemporaneous shocks can equally well serve as instruments for our moment conditions. For example, we can use high-frequency monetary shocks as instruments even if these are in fact an amalgam of true interest rate surprises and information shocks \citep{BauerSwanson2023}.

Our method allows researchers to freely choose which cross-sectional moments of agents' choices to fit. These moments could be simple, such as the mean, variance, or interquartile range of the consumption distribution, or more complicated, such as the share of illiquid savings held by households above the 90th percentile of the wealth distribution.

Of course, by focusing on a single model block, our method does not permit estimation of arbitrary general equilibrium counterfactuals; only those that are functions of parameters within the estimated block. This may suffice for some questions: in some HANK models the steady-state Lorenz curve for wealth, say, is a function only of parameters in the household block. Moreover, our method allows researchers to test the validity of individual blocks one by one. Finally, the limited-information parameter estimates can serve as a robust benchmark for subsequent full-model calibrations.

In addition to the assumptions already discussed, our method shares the following requirements with existing full-information estimation procedures: the model block of interest must be correctly specified (this assumption is testable within our framework); the linearization of micro aggregates with respect to macro variables must be accurate, as in \citet{Reiter2009} and \citetalias{AuclertBardoczyRognlieStraub2021}; and agents must have rational expectations. For computational feasibility, we rely on the rapid and numerically stable software for computing SSJs developed by \citetalias{AuclertBardoczyRognlieStraub2021}.\footnote{Specifically, we use their excellent {\tt sequence-jacobian} Python package off the shelf.}

We illustrate the empirical usefulness of our approach by applying it to data simulated from an empirically calibrated two-asset HANK model (\citealp{KaplanMollViolante2018}; \citetalias{AuclertBardoczyRognlieStraub2021}). The estimation procedure is fed time series data on portfolio returns, aggregate earnings, moments of the cross-sectional distributions of household consumption and savings, and monetary and fiscal shocks. The procedure exploits only the structure of the model's heterogeneous household block, up to several unknown parameters; it does not place any restrictions on production, corporate investment, financial intermediation, monetary policy, trade, etc. We estimate parameters related to household preferences and portfolio adjustment costs. Even in a moderate sample size of 120 quarters, the parameter estimators have only modest mean squared error, and the bootstrap confidence intervals and model specification test have accurate coverage/size.

Econometrically, we express our estimator as a minimum distance procedure in the frequency domain. This facilitates econometric analysis of the infinite decision horizon of the agents in the model, and it also suggests a natural and effective bootstrap procedure. We prove consistency and asymptotic normality for a general class of ``spectral GMM'' estimators under weak assumptions on the time series dependence of the data. Exploiting the bilinear form of our minimum distance moment function, we are able to relax the existing assumptions in the literature used to prove stochastic equicontinuity of the relevant empirical spectral means.

\paragraph{Literature.}
Our approach is most closely related to the moment-matching and likelihood estimation procedures for HA models cited earlier. All these procedures require a full specification of general equilibrium. We drop this requirement, at the expense of having to observe the macro sufficient statistics and a set of identified shocks. While superficially similar, impulse response matching in the tradition of \citet{RotembergWoodford1997} generally requires a full equilibrium specification (though non-matched shocks need not be specified); our strategy instead considers a \emph{combination} of impulse responses that is provably invariant with respect to the auxiliary model blocks.

Though the classic literature on estimating consumption Euler equations from micro data has wrestled with challenges arising from borrowing constraints \citep[e.g.,][]{Deaton1985,Hayashi1985,Carroll2001}, none of those estimation procedures appear to be valid under the kinds of data generating processes considered in the recent HANK literature. By contrast, our method can handle a general class of models with essential nonlinearities due to occasionally binding constraints, possibly latent micro and macro state variables, and portfolio adjustment costs. However, unlike the earlier literature, our method's ability to exploit micro panel data is limited (see \cref{sec:extensions}).

Our paper is concerned with structural estimation of a HA model block, unlike the semi-structural modeling approaches of \citet{ChangChenSchorfheide2024,EttmeierKimSchorfheide2024,AlmuzaraArellanoBlundellBonhomme2025,ChangKimPark2025}. In these papers, the economic structure is primarily used to identify shocks and characterize the income process, rather than to estimate the deep structural parameters. By contrast, we estimate these parameters directly. The benefits of structural estimation are well known: parameters have deep economic meaning, we can compute well-defined counterfactuals, and we can directly test the validity of concrete theoretical mechanisms. The downside is that our estimation results are sensitive to misspecification of the model block of interest (but not to misspecification of other aspects of the general equilibrium).

A few papers have taken application-specific approaches that are similar in spirit to ours. \citet{GagliardoneGertlerLenzuTielens2025} estimate the New Keynesian Phillips Curve from individual-firm optimality conditions, exploiting data on the ``sufficient statistics'' for the firm, namely marginal costs. \citet{DebortoliGali2025} estimate linear consumption Euler equations motivated by HA models with hand-to-mouth households. \citet{Tryphonides2023} partially identifies the preference parameters in a HA consumption block by deriving model-implied moment inequalities; these require particular micro variables to be observed by the econometrician. Unlike these papers, our framework covers a general class of HA models with potentially non-smooth individual decision rules, and all macro and micro variables in the decision problem are allowed to be latent except the vector of aggregate sufficient statistics.

\paragraph{Outline.}
\cref{sec:example} gives an overview of our procedures in the context of estimating a heterogeneous household model block. \cref{sec:framework} defines our general estimation and inference procedures, and \cref{sec:theory} states formal econometric results on their asymptotic properties. \cref{sec:simul} applies our methods to data simulated from a two-asset HANK model. \cref{sec:concl} concludes. Proofs and technical details are relegated to the appendix. %A code repository with general Python functions and replication details is available on GitHub.\footnote{URL XX}

\section{Motivating example}
\label{sec:example}

Here we describe the intuition behind our estimation approach in the context of the model that we will use in our empirically calibrated simulation study in \cref{sec:simul}: a two-asset heterogeneous household model block.

\subsection{Two-asset heterogeneous household model}
\label{sec:model}

We consider the household block of the two-asset HANK model of \citetalias{AuclertBardoczyRognlieStraub2021} (Appendix B.3), which in turn is closely related to the model of \citet{KaplanMollViolante2018}. There is a continuum of households $i \in [0,1]$, each of which allocates their savings at discrete time $t$ between liquid assets $b_{i,t}$ and illiquid assets $a_{i,t}$. Changes in illiquid asset holdings incur a convex portfolio cost $\Psi(a_{i,t},(1+r_t^a)a_{i,t-1})$. Hours worked $N_t$ is uniform across households and determined by a labor union that is outside this model block and thus unrestricted. The household solves the problem
\[\max_{\lbrace c_{i,t},b_{i,t},a_{i,t} \rbrace_{t=0}^\infty} E_0\sum_{t=0}^\infty \beta^t \frac{c_{i,t}^{1-1/\gamma}}{1-1/\gamma}\]
\vspace{-1.5\baselineskip}
\begin{align*}
\text{s.t.}\quad & c_{i,t}+a_{i,t}+b_{i,t} = (1-\tau_t)w_t N_t e_{i,t} + (1+r_t^b)b_{i,t-1} + (1+r_t^a)a_{i,t-1} - \Psi(a_{i,t},(1+r_t^a)a_{i,t-1}), \\
& a_{i,t} \geq 0,\quad b_{i,t} \geq \underline{b}.
\end{align*}
Here $E_0$ denotes the expectation at time 0, $\beta$ is the discount factor, $\gamma$ is the elasticity of intertemporal substitution (EIS), $\tau_t$ is the labor tax, $w_t$ is the aggregate real wage, $r_t^b$ and $r_t^a$ are the real returns on liquid and illiquid wealth, respectively, and $\underline{b}$ is the borrowing constraint for liquid assets. $e_{i,t}$ is an idiosyncratic household productivity level, which evolves according to a Markov process, independently of any macro shocks.

Note that in this model, households' optimal choices depend on the macroeconomy only through three aggregate variables: after-tax labor earnings $(1-\tau_t)w_tN_t$, and the two returns $r_t^b,r_t^a$. We collect these three variables in the vector $x_t \equiv ((1-\tau_t)w_tN_t,r_t^b,r_t^a)'$ of macro ``sufficient statistics''. The household's time-$t$ policy functions for consumption and asset allocations are functions of the current idiosyncratic states $(b_{i,t-1},a_{i,t-1},e_{i,t})$ as well as the household's beliefs about all the current and future values of the sufficient statistics: $(x_t,x_{t+1},x_{t+2},\dots)$. Crucially---and unlike existing full-information estimation procedures---we do not place any functional form restrictions on the dynamics of $x_t$.

We seek to estimate the vector $\theta$ of structural parameters for this model block. The vector includes the discount factor $\beta$, the EIS $\gamma$, any parameters that enter into the portfolio adjustment cost function $\Psi$, and any parameters that govern the dynamics of the idiosyncratic exogenous state $e_{i,t}$.

\subsection{Sequence-space representation}
\label{sec:ssj}

Having defined the household model block, we next describe the sequence-space representation of the aggregated household decisions. This is the most important step in deriving the moment conditions that we will take to the data.

Let $y_{i,t}$ be a vector of ``outputs'' of the household's decision problem, that is, functions of its choice variables $(c_{i,t},b_{i,t},a_{i,t})$, idiosyncratic states $(b_{i,t-1},a_{i,t-1},e_{i,t})$, and macro variables $x_t$. Suppose the econometrician observes the cross-sectional aggregate vector $y_t = \int_0^1 y_{i,t}\, di$. For example, if $y_{i,t} = (c_{i,t},c_{i,t}^2)'$, then $y_t$ equals the cross-sectional mean and second moment of consumption, which may be measured empirically using repeated cross sections from a household expenditure survey (measurement error is discussed below). Assume that all households share the same beliefs about the evolution of the macro sufficient statistics $x_t$. Then the aggregated outputs $y_t$ are functions of two objects: households' beliefs about $(x_t,x_{t+1},x_{t+2},\dots)$, and the pre-determined cross-sectional distribution $\mu_{t-1}$ of idiosyncratic states $(b_{i,t-1},a_{i,t-1},e_{i,t})$.\footnote{The distribution is pre-determined in the sense that it only depends on macro shocks up to date $t-1$.} Both these objects are infinite-dimensional.

While the fully nonlinear solution to the model block is computationally intractable, linearizing with respect to the macro variables restores tractability. This linearization approach, which preserves the full \emph{micro} heterogeneity, follows \citet{Reiter2009} and a large subsequent literature, and it has been the basis of most full-information estimation procedures for HA models. As in the general frameworks of \citet{BoppartKrusellMitman2018} and \citetalias{AuclertBardoczyRognlieStraub2021}, linearization of the aggregate household outputs $y_t$ with respect to $x_t$ around a deterministic steady state (denoted with superscript ``ss'') yields an approximately linear \emph{sequence-space} representation
\begin{equation} \label{eqn:ss_rep}
y_t-y^\text{ss} \approx E_t\left[\sum_{k=0}^\infty J_k(\theta) (x_{t+k}-x^\text{ss})\right] + \xi_{t-1}(\theta).
\end{equation}
Here $E_t$ is the households' expectation operator at time $t$, and $\xi_{t-1}(\theta)$ is a pre-determined term capturing the effects of the cross-sectional state distribution $\mu_{t-1}$ on current household choices. The sequence $\lbrace J_k(\theta) \rbrace_{k=0}^\infty$ of parameter-dependent matrices is called the \emph{sequence-space Jacobians} (SSJs). These are the central ingredients in our estimation procedure, as they measure how the future expected time paths of macro sufficient statistics translate into aggregated household choices.\footnote{Note that idiosyncratic risk and precautionary behavior are embedded in the steady-state distribution and SSJs. Aggregate risk is not captured directly at first order. Incorporating aggregate risk effects would generally require a higher-order expansion and is beyond the scope of this paper.} \citetalias{AuclertBardoczyRognlieStraub2021} developed a ``fake news algorithm'' which rapidly and reliably computes the whole sequence of SSJs for any given structural parameters $\theta$. In practice, their algorithm truncates the households' decision horizon at a finite number $H$, yielding a finite sum in \eqref{eqn:ss_rep}. $H$ can be taken to be on the order of 500 to minimize truncation bias.

While the representation \eqref{eqn:ss_rep} is merely a first-order Taylor expansion that leaves out terms of order 2 and higher in the macro variables, we follow most of the literature in ignoring the approximation error. As argued originally by \citet{Reiter2009}, the linearization exploits the fact that, whereas individual-household decision rules are non-smooth functions of their own idiosyncratic states as well as of the macro sufficient statistics, \emph{integrals} of these choices across agents are often smooth functions of the sufficient statistics. If the integrated choices are smooth functions of $\lbrace E_t[x_{t+k}] \rbrace_k$, and the underlying macro shocks driving $x_t$ are not very large, the first-order linearization error will be negligible. We therefore henceforth treat \eqref{eqn:ss_rep} as an exact equality, as in \citetalias{AuclertBardoczyRognlieStraub2021} as well as the voluminous literature on estimating representative agent models via the Kalman filter likelihood.

Despite having access to an efficient numerical algorithm for computing the SSJs, there are two challenges involved in taking representation \eqref{eqn:ss_rep} to the data. First, in an ideal setting, we would observe exogenous variation in the household expectations $E_t[x_{t+k}]$ at all future horizons $k \geq 1$, in which case we could estimate $\theta$ through nonlinear least squares. However, in practice we do not directly observe households' expectations at all forecast horizons.\footnote{In representative agent models, it is often possible to derive estimating equations that involve only the one-step-ahead expectation (e.g., the conventional consumption Euler equation or New Keynesian Phillips Curve). One could then directly measure this expectation using survey data \citep[e.g.,][]{Fuhrer2017}. However, we are not aware of a method that can be applied to general HA models that would eliminate expectations at intermediate and long horizons from the households' optimality conditions.} Second, the pre-determined term $\xi_{t-1}(\theta)$ is generally a complicated function of the cross-sectional state distribution $\mu_{t-1}$ and is therefore not directly observable either; yet it is likely correlated with the macro sufficient statistics $x_t$ and so cannot be treated as a harmless error term. Our approach addresses both challenges by using identified shocks as instruments.

The full-information estimation method of \citetalias{AuclertBardoczyRognlieStraub2021} goes on to specify a complete equilibrium structure and a full set of shock processes. This allows those authors to compute the SSJs of all endogenous macro variables with respect to all exogenous shocks, from which they can derive the likelihood function for any vector of observed macro time series in the model. The full-information approach, while highly informative when correctly specified, implicitly places strong functional form restrictions on the dynamics of the macro sufficient statistics $x_t$. If we only care about estimating and testing the household model block, we now show that we can avoid (a) placing any functional form restrictions on the dynamics of $x_t$ and (b) specifying the full list of shocks driving the data.

\subsection{Shocks as instruments}
\label{sec:iv}
We are able to derive empirically useful limited-information moment conditions if we observe an appropriate set of instruments and impose rational expectations. Suppose we have access to a vector $z_t$ of instruments satisfying the following two conditions:
\begin{enumerate}
	\item The instruments are orthogonal to past shocks, and therefore orthogonal to the pre-determined term:\footnote{In the first-order linearization, $\xi_{t-1}(\theta)$ is linear in past shocks.}
	\begin{equation} \label{eqn:iv_orthog}
	\cov(z_t, \xi_{t-1}(\theta)) = 0.
	\end{equation}
	\item The instruments $z_t$ do not predict households' forecast errors at any horizon:
	\begin{equation} \label{eqn:forec_err_orthog}
	\cov(x_{t+k} - E_t[x_{t+k}],z_t) = 0\quad \text{for}\quad k=0,1,2,\dots
	\end{equation}
	A set of sufficient conditions for this second assumption is that (i) households have rational expectations and (ii) $z_t$ is contained in the households' time-$t$ information set. These conditions are often imposed in full-information estimation.
\end{enumerate}
Combining assumptions \eqref{eqn:iv_orthog}--\eqref{eqn:forec_err_orthog} with representation \eqref{eqn:ss_rep} (viewed as an equality), we obtain a system of moment conditions:
\begin{equation} \label{eqn:moment_cond}
\cov(y_t,z_t) = \sum_{k=0}^\infty J_k(\theta) \cov(x_{t+k},z_t).
\end{equation}
Assuming that $(y_t,x_t,z_t)$ are all observed, this is a system of instrumental variable exogeneity conditions, but with two non-standard features: the infinite horizon of the sum, and the nonlinearity of the coefficient matrices $J_k(\theta)$ in the structural parameters $\theta$.

Identified economic shocks are prime candidates for instruments $z_t$. A wide range of shock series are available in the applied literature, such as monetary shocks, fiscal shocks, oil shocks, and technology shocks; see the lists compiled by \citet{Ramey2016} and \citet{BarnichonMesters2020}. By virtue of being ``shocks'', these series ought to be unpredictable and thus satisfy the first assumption \eqref{eqn:iv_orthog} on the instruments. Moreover, if the shocks are economically important, they are also likely to be observed by households, thus satisfying the second assumption \eqref{eqn:forec_err_orthog} on the instruments under rational expectations. Unlike in the case of full-information likelihood estimation, (a) we do not need to take a stand on which and how many shocks are driving the data, and (b) our moment conditions \eqref{eqn:moment_cond} allow the observed data to be contaminated by a general class of measurement error processes.\footnote{Specifically, the moment conditions continue to hold for the contaminated processes $y_t^\text{obs} = y_t + \epsilon_t^y$, $x_t^\text{obs} = x_t + \epsilon_t^x$, $z_t^\text{obs} = z_t + \epsilon_t^z$, provided that (a) $\lbrace z_t \rbrace$ is dynamically uncorrelated with $\lbrace \epsilon_t^y,\epsilon_t^x \rbrace$ and (b) $\lbrace \epsilon_t^z \rbrace$ is dynamically uncorrelated with $\lbrace y_t,x_t,\epsilon_t^y,\epsilon_t^x \rbrace$.}

Conveniently, the instruments continue to satisfy our assumptions if they are given by \emph{functions} of several true underlying economic shocks, i.e., $z_t = \mathrm{fct}(\varepsilon_{1t},\varepsilon_{2t},\dots)$. For example, even if a measured ``monetary shock'' instrument is actually an amalgam of a pure interest rate surprise with other kinds of shocks (such as information shocks), it will still satisfy our assumptions.

Assuming valid instruments, identification of the structural parameters requires the moment conditions \eqref{eqn:moment_cond} to have a unique solution $\theta$. It is necessary that the total number of moments $\dim(y_t) \times \dim(z_t)$ weakly exceed the number of parameters $\dim(\theta)$. Additionally, the instruments $z_t$ must induce sufficient variation in current and future sufficient statistics $x_{t+k}$ in order for the right-hand side of \eqref{eqn:moment_cond} to be a non-trivial function of the structural parameters $\theta$---at a minimum, $z_t$ must be correlated with $x_{t+k}$ at \emph{some} horizon $k$, an easily testable condition. As a heuristic example, if $\theta$ includes the household discount factor, $z_t$ must co-vary with $x_{t+k}$ at relatively large horizons $k$, so as to be able to distinguish the choices $y_t$ made by patient and impatient households. We give precise conditions for identification in \cref{sec:identification} below.

\subsection{Estimator and over-identification test}
\label{sec:estimator_intuit}
Given a data set $\lbrace y_t,x_t,z_t \rbrace_{t=1}^T$, our estimator $\widehat{\theta}$ of the structural parameters satisfies the sample analogue of the population moment conditions \eqref{eqn:moment_cond} as well as possible:
\begin{equation} \label{eqn:sample_moment_approx}
\widehat{\cov}(y_t,z_t) \approx \sum_{k=0}^{T-1} J_k(\widehat{\theta}) \widehat{\cov}(x_{t+k},z_t),
\end{equation}
where ``$\widehat{\cov}$'' denotes sample covariance. Notice that we truncate the sum at horizon $k=T-1$, since the data is uninformative about covariances at longer horizons.\footnote{As previously discussed, the SSJs $J_k(\theta)$ are also truncated at a large horizon $k \leq H$ for computational feasibility, but it is computationally feasible to choose $H \gg T$.} When the model is over-identified---i.e., $\dim(y_t) \times \dim(z_t) > \dim(\theta)$---it is typically impossible to satisfy the above system of equations with equality in a finite sample, so we define $\widehat{\theta}$ as a minimum distance estimator that minimizes a weighted quadratic form in the discrepancies between the left-hand and right-hand sides; see \cref{sec:estimator} for details.

The sample estimating equation \eqref{eqn:sample_moment_approx} shows that our estimation approach can be viewed as a nonlinear ``regression in impulse response space'', using the SSJs as the key link between data and model. The left-hand side of \eqref{eqn:sample_moment_approx} is the contemporaneous empirical impulse response of the aggregate household outputs $y_t$ with respect to the instruments $z_t$, while the right-hand side involves the empirical impulse responses $\widehat{\cov}(x_{t+k},z_t)$ of the macro sufficient statistics with respect to the instruments at all horizons $k=0,1,\dots,T-1$. This estimation approach, which is applicable to general nonlinear HA models, builds on the linear, representative agent approach of \citet{BarnichonMesters2020}.

In over-identified settings, we can test whether the household model block is correctly specified. This can be done by checking whether a single parameter vector $\widehat{\theta}$ can approximately satisfy the entire system of sample estimating equations \eqref{eqn:sample_moment_approx}. More specifically, we carry out a conventional over-identification test for minimum distance estimators; see \cref{sec:overid} below.

\subsection{Summary and outlook}
To recap, our estimation procedure is a nonlinear regression in impulse response function space that exploits SSJs as the key theoretical link between (a) the responses of macro sufficient statistics and (b) the responses of aggregated household choices. For computational feasibility, our method relies on existing software for computing SSJs \citepalias{AuclertBardoczyRognlieStraub2021}. For econometric validity, our method relies on five key assumptions. First, the household block is correctly specified (this is testable in the over-identified case). Second, the household decision problem depends on the macroeconomy only through a finite-dimensional vector of ``sufficient statistics''. Third, the approximation error from linearizing the aggregated optimal choices of households with respect to macro variables is negligible. Fourth, households have rational expectations. Fifth, we observe time series data on all the macro sufficient statistics as well as a vector of instruments, with the latter being functions of contemporaneous shocks. The first four of these assumptions are shared by most existing full-information estimation methods.

Our method gives the researcher flexibility to choose the moments $y_t$ of households' choices (e.g., cross-sectional moments of consumption and savings) and the instruments $z_t$ (e.g., identified shocks). It is not necessary to obtain data on any of the households' idiosyncratic state variables---if such data happened to be available, it could be incorporated in $y_t$.

In the following sections we will argue that the above estimation and testing approach is applicable to a large class of HA model blocks, including blocks with heterogeneous firms, financial intermediaries, etc. Moreover, we will provide general identification results, prove the consistency and asymptotic normality of the minimum distance estimator $\widehat{\theta}$ under weak conditions, and propose analytical and bootstrap procedures for computing standard errors and confidence intervals. The econometric analysis is non-trivial as we allow for general serial correlation patterns in the data and explicitly handle the error imparted by the truncation of the infinite-horizon sum in the population moment conditions \eqref{eqn:moment_cond}.

\section{General framework}
\label{sec:framework}

Having explained the intuition behind our approach, we now formally define our general framework and econometric procedures.

\subsection{Moment conditions}
We take the moment conditions \eqref{eqn:moment_cond} derived in the previous section as the starting point for our general analysis. We restate those conditions below, using the notation $\theta_0$ to denote the true value of the structural parameter vector.

\begin{asn}[Moment conditions] \label{asn:moment_cond}
There exist a sequence of deterministic, matrix-valued functions $J_k \colon \Theta \to \mathbb{R}^{d_y \times d_x}$ for $k=0,1,\dots$ and a parameter vector $\theta_0 \in \Theta \subset \mathbb{R}^{d_\theta}$ such that
\[\cov(y_t,z_t) = \sum_{k=0}^\infty J_k(\theta_0) \cov(x_{t+k},z_t),\]
where $y_t,x_t,z_t$ are jointly stationary random vectors of dimensions $d_y,d_x,d_z$, respectively.
\end{asn}
The theoretical analysis in \cref{sec:theory} will impose further assumptions on the SSJs and the data to ensure that both sides of the above equation are mathematically well-defined.

We derived moment conditions of the above form for a particular heterogeneous household model in \cref{sec:example}, but it is clear from our arguments that such conditions will equally well obtain for any other HA model block covered by the general framework of \citetalias{AuclertBardoczyRognlieStraub2021}. The only differences between models are (a) which variables enter into the macro sufficient statistics $x_t$, (b) the functional form of the SSJs $J_k(\theta)$, and (c) which aggregated agent choices $y_t$ are applicable; but given these, the system of equations in \cref{asn:moment_cond} will be satisfied (ignoring linearization error).

In particular, model blocks with heterogeneous firms or financial intermediaries also have a sequence-space representation which, when combined with valid instruments as discussed in \cref{sec:iv}, will result in moment conditions of the form in \cref{asn:moment_cond}, for particular model-specific choices of $x_t$ and $\lbrace J_k(\theta) \rbrace_k$. See for example \citet{WinberryAuclertRognlieStraub2025} for the case of heterogeneous firms.\footnote{While our focus here is on traditional applications in macroeconomics, one could imagine applying the sequence-space machinery to models from structural microeconomics where agents' choices depend on aggregate variables, such as models of schooling choice given time-varying skill premia.}

\subsection{Estimator}
\label{sec:estimator}

We now define the minimum distance estimator. Rather than formulating the sample moment conditions in the time domain as in \eqref{eqn:sample_moment_approx}, we define the estimator in the frequency domain. This will turn out to both facilitate the theoretical econometric analysis in \cref{sec:theory} and motivate an effective bootstrap strategy for inference. To that end, let
\[S_{ab}(\omega) \equiv \frac{1}{2\pi}\sum_{h=-\infty}^\infty e^{-\iota \omega h} \cov(a_{t+h},b_t)\]
denote the cross-spectral density matrix for arbitrary stationary time series $a_t,b_t$ (with absolutely summable autocovariance function), where $\iota =\sqrt{-1}$. Then the moment conditions in \cref{asn:moment_cond} can be equivalently stated as
\begin{eqnarray} \label{eqn:moment_cond_spec}
\int_0^{2\pi} \left\lbrace S_{yz}(\omega) - J(\omega;\theta_0) S_{xz}(\omega) \right\rbrace\,d\omega = 0_{d_y \times d_z},
\end{eqnarray}
where $J(\omega;\theta) \equiv \sum_{k=0}^\infty e^{\iota k \omega} J_k(\theta)$. The standard sample analogue of $S_{ab}(\omega)$ is the cross-periodogram $\widehat{S}_{ab}(\omega) \equiv \frac{1}{2\pi T}\lbrace\sum_{t=1}^T e^{-\iota\omega t}(a_t-\bar{a})\rbrace \lbrace \sum_{t=1}^T e^{\iota\omega t}(b_t-\bar{b})' \rbrace$, where $\bar{a}\equiv\frac{1}{T}\sum_{t=1}^T a_t$ and similarly for $\bar{b}$. The following minimum distance estimator simply plugs the periodogram into the population moment condition \eqref{eqn:moment_cond_spec} and approximates the integral with a sum. Let $d_g \equiv d_yd_z$ denote the total number of moments.
\begin{defn}[Estimator]
Given a symmetric positive definite weight matrix $\widehat{W} \in \mathbb{R}^{d_g \times d_g}$, the minimum distance estimator is defined as
\[\widehat{\theta} \equiv \argmin_{\theta \in \Theta}\; \widehat{g}(\theta)'\widehat{W}\widehat{g}(\theta),\]
with moment function
\[\widehat{g}(\theta) \equiv \frac{2\pi}{T}\sum_{j=0}^{T-1} \ve\left\lbrace \widehat{S}_{yz}(\omega_j) - J(\omega_j;\theta) \widehat{S}_{xz}(\omega_j) \right\rbrace,\quad \text{where}\quad \omega_j \equiv \frac{2\pi j}{T}.\]
\end{defn}
In practice, we use a gradient-based numerical optimization routine with multiple starting values to compute the minimum. As is well known, the cross-periodogram can be computed efficiently using the discrete Fourier transform. Similarly, for any $\theta$, given the SSJs $\lbrace J_k(\theta) \rbrace_k$ produced by the \citetalias{AuclertBardoczyRognlieStraub2021} algorithm, the spectral SSJ function $J(\omega;\theta)$ can be computed rapidly by another (inverse) Fourier transform.\footnote{As discussed in \cref{sec:ssj}, in practice SSJs are computed using a finite truncation horizon $H$. Here we abstract from this minor numerical issue and assume $H=\infty$.} As defined, the moment function $\widehat{g}(\theta)$ is theoretically guaranteed to be a real vector for any $\theta$; however, numerical errors can cause the computed $\widehat{g}(\theta)$ to have a very small but nonzero imaginary part, so in practice we just retain the real part.

\subsection{Inference}
\label{sec:inference}
In \cref{sec:consistency_AN} we will show that the minimum distance estimator is consistent and asymptotically normal under weak conditions. Its asymptotic variance-covariance matrix can be estimated consistently by
\[\widehat{\Sigma} \equiv (\widehat{G}'\widehat{W}\widehat{G})^{-1}\widehat{G}'\widehat{W}\widehat{\Omega}\widehat{W}\widehat{G}(\widehat{G}'\widehat{W}\widehat{G})^{-1},\]
where
\[\widehat{G} \equiv \frac{\partial \widehat{g}(\widehat{\theta})}{\partial \theta'} \in \mathbb{R}^{d_g \times d_\theta},\]
\[\widehat{\Omega} \equiv \sum_{|h| \leq L_T} K\left(\frac{h}{L_T}\right)\widehat{\Gamma}_\psi(h) \in \mathbb{R}^{d_g\times d_g},\quad \widehat{\Gamma}_\psi(h) \equiv \begin{cases}
	\frac{1}{T}\sum_{t=h+1}^T \left(\widehat{\psi}_t-\overline{\widehat{\psi}}\right)\left(\widehat{\psi}_{t-h}-\overline{\widehat{\psi}}\right)', & h \geq 0, \\
	\widehat{\Gamma}_\psi(-h)', & h<0,
\end{cases}\]
\[\widehat{\psi}_t \equiv (z_t - \bar{z}) \otimes \left(y_t-\bar{y} - \sum_{k=0}^{T-t} J_k(\widehat{\theta})(x_{t+k}-\bar{x})\right) \in \mathbb{R}^{d_g},\quad \overline{\widehat{\psi}} \equiv \frac{1}{T}\sum_{t=1}^T \widehat\psi_t,\]
$\bar{y} \equiv T^{-1}\sum_{t=1}^T y_t$, $\bar{x}$, and $\bar{z}$ are sample averages, and $K(\cdot)$ is a heteroskedasticity and autocorrelation consistent (HAC) kernel function satisfying the standard conditions in \cref{app:hac}.\footnote{$\widehat{\psi}_t$ is an approximation to the per-observation score process for the sample moment function \eqref{eqn:sample_moment_approx}, truncating the sum at the longest horizon available in the data.} The derivative matrix $\widehat{G}$ can be computed via a finite difference approximation. As usual, standard errors for the individual parameter estimates $\widehat{\theta}_j$ can be computed as $\sqrt{\widehat{\Sigma}_{jj}/T}$, $j=1,\dots,d_\theta$. In the simulation study in \cref{sec:simul}, we use the Newey-West HAC kernel $K(u) = \max\lbrace 1-|u|,0 \rbrace$ with bandwidth selection rule $L_T = \lceil 2.24 \times T^{1/3} \rceil$.\footnote{This bandwidth rule is mean-squared-error-optimal if the true score process is an AR(1) process with $\rho=0.7$ \citep[equations 6.2 and 6.4]{Andrews1991}.}

While the variance estimator above is computationally attractive, it is based on an asymptotic delta method linearization that can be inaccurate in finite samples whenever the SSJs are highly nonlinear in $\theta$. Our simulation study in \cref{sec:simul} finds that a bootstrap inference procedure reliably delivers confidence intervals with accurate coverage. This procedure resamples the periodogram using a Gaussian multiplier bootstrap and re-runs the nonlinear optimization of the objective function on each bootstrap data set. See \cref{app:bootstrap} for details on the algorithm.

\subsection{Over-identification test}
\label{sec:overid}
In the over-identified case $d_g>d_\theta$, a test of the specification of the model block can be carried out by computing the statistic
\[\widehat{\Upsilon} \equiv T\widehat{g}(\widehat{\theta})'\left\lbrace \widehat{\Omega}^{-1} - \widehat{\Omega}^{-1}\widehat{G}(\widehat{G}'\widehat{\Omega}^{-1}\widehat{G})^{-1}\widehat{G}'\widehat{\Omega}^{-1} \right\rbrace \widehat{g}(\widehat{\theta}),\]
see \citet[Section 9.5]{NeweyMcFadden1994}.\footnote{Note that the formula plugs in the HAC estimate $\widehat{\Omega}$ regardless of which weight matrix $\widehat{W}$ is used to compute $\widehat{\theta}$.} Intuitively, this statistic checks whether the single parameter vector $\widehat{\theta}$ is able to approximately satisfy all moment conditions simultaneously. Under the joint null hypothesis of correct specification of the model block and instrument validity, the statistic has an asymptotic chi-squared distribution with $d_g-d_\theta$ degrees of freedom. However, our simulation study in \cref{sec:simul} suggests that the use of analytical critical values can cause size distortions in finite samples, whereas a bootstrap critical value controls size well even in modest sample sizes. The bootstrap procedure is described in \cref{app:bootstrap}.

\subsection{Weight matrix}
\label{sec:weight_mat}
A natural, though \emph{ad hoc}, choice of weight matrix is $\widehat{W} = \widehat{V}_z^{-1} \otimes \widehat{V}_y^{-1}$, where $\widehat{V}_z$ is a diagonal matrix with the sample variances of the elements of $z_t$ on the diagonal, while $\widehat{V}_y$ similarly collects the sample variances of $y_t$ on its diagonal. This choice of weight matrix makes the entire estimation procedure invariant to the units of $y_t$ and $z_t$.

Asymptotically, the efficient choice of weight matrix can be consistently estimated by $\widehat{W}^\text{eff}=\widehat{\Omega}^{-1}$, where $\widehat{\Omega}$ is obtained using any initial positive definite weight matrix, such as the one suggested above \citep[Section 5.2]{NeweyMcFadden1994}. However, in some applications, the number $d_g$ of moments may be large relative to the sample size $T$, likely causing the purported ``efficient'' estimator to be poorly behaved in small samples \citep{AltonjiSegal1996}. We recommend that users run Monte Carlo studies calibrated to their application to compare the finite-sample performance of estimators with different choices of weight matrix.

\subsection{Richer agent choice data}
\label{sec:extensions}
Our procedure can exploit richer data on agents' choices than simple cross-sectional moments. So far we have restricted attention to aggregated outputs $y_t$ of the agents' decisions that can be written as integrals $y_t = \int_0^1 y_{i,t}\,di$, where $y_{i,t}$ is a vector of simple transformations of households' choice variables. However, as emphasized by \citetalias{AuclertBardoczyRognlieStraub2021} (Appendix A), the linearized representation \eqref{eqn:ss_rep} will also obtain for many other functionals $y_t$ of the cross-sectional distribution of agent choices, such as quantiles, conditional expectations, etc. Their Python package readily computes SSJs for such functionals.\footnote{Mechanically, in their {\tt sequence-jacobian} Python package, any functional of the joint cross-sectional distribution of agent states and choices can be attached to the model block using the {\tt add\_hetoutputs} function. The code will then automatically compute the SSJs of these functionals with respect to the macro sufficient statistics, just like it would for any other aggregate variables.} As a concrete example, in the simulation study in \cref{sec:simul}, we consider cross-sectional moments of the form $\int a_{i,t} \1\lbrace q_{p_0,t} \leq  b_{i,t}+a_{i,t} \leq q_{p_1,t} \rbrace\,di$, i.e., the total illiquid assets held by households who are between the $100p_0$ and $100p_1$ percentiles of the wealth distribution. By looking at various quantiles $p_0,p_1$, these moments give detailed information on the dynamics of wealth inequality, which should provide identifying power for many of the parameters of the household problem.

Our method can also exploit short micro panels, though in limited ways. The Python package of \citetalias{AuclertBardoczyRognlieStraub2021} easily allows users to compute cross-sectional moments of the entire joint distribution of agent states and choices. Some of these state variables may be lagged choice variables, such as asset holdings. For example, in the model from \cref{sec:model}, it would be straightforward to compute SSJs for any smooth moments of the form $y_t = \int_0^1 \mathrm{fct}(a_{i,t-1},b_{i,t-1},a_{i,t},b_{i,t},c_{i,t},x_t)\,di$. This allows researchers to at least exploit two-period rotating panels on liquid and illiquid asset holdings. It would be an interesting question for future research to develop an SSJ algorithm for aggregated outcomes that involve micro data for more than two consecutive time periods.

\section{Theoretical results}
\label{sec:theory}
This section presents our formal econometric assumptions and proves identification, consistency, asymptotic normality, and consistent estimation of the asymptotic variance. We encourage applied readers to skip ahead to the next section on their first reading.

The proofs in \cref{app:proofs} in fact prove consistency and asymptotic normality for a more general class of ``spectral GMM'' estimators; the results in the present section are special cases, as formally verified in \cref{sec:specialization-our-setup}. In addition to potentially opening the door to other kinds of applications, the general framework arguably makes the proofs more transparent.

\subsection{Main assumptions}
Let the observed data be denoted $\zeta_t\equiv(y_t',x_t',z_t')'\in\mathbb R^{d_\zeta}$ with $d_\zeta\equiv d_y+d_x+d_z$. All norms below are Euclidean for vectors or Frobenius for matrices, and the dimensions are fixed so norms are equivalent.

\begin{asn}[Data]
\label{asn:stationarity}
\begin{asnitem}
\item \label{asn:stationarity:i} $\zeta_t$ is strictly stationary with mean zero and finite fourth moments.
\item \label{asn:stationarity:ii} The joint cumulants of orders $\ell \in \lbrace 2,4 \rbrace$ are absolutely summable:\footnote{We interpret the norm as applying to the vectorization of the cumulant tensor.}
\[
\sum_{h_1\in\mathbb Z}\sum_{h_2\in\mathbb Z}\cdots\sum_{h_{\ell-1}\in\mathbb Z}
\left\|\cum\left(\zeta_0,\zeta_{h_1},\dots,\zeta_{h_{\ell-1}}\right)\right\|<\infty.
\]
\end{asnitem}
\end{asn}
Recall that \cref{asn:moment_cond} already imposed stationarity of $\lbrace \zeta_t \rbrace$. The mean zero assumption above is just to simplify notation; in practice, we work with de-meaned data.\footnote{The difference between population and sample demeaning affects only the zero Fourier ordinate. As a result, the change in the sample moments $\widehat{g}(\theta)$ and their derivative induced by sample de-meaning is $O_p(T^{-1})$ uniformly in $\theta$ under our assumptions, and is therefore asymptotically negligible.} The summability assumptions on the second- and fourth-order cumulants follow \citet{Brillinger1981} and \citet{Andrews1991}. These conditions are for example implied by mixing assumptions and hold for many stationary time series models; see \citet{Andrews1991}.

\begin{asn}[Parameter space]
\label{asn:parameter-space}
$\Theta\subset\mathbb R^{d_\theta}$ is compact, and the true parameter $\theta_0\in\Theta$.
\end{asn}

The compactness assumption is conventional, but could be relaxed under shape restrictions on the SSJs.

\begin{asn}[SSJ]
\label{asn:jacobian}
\begin{asnitem}
\item \label{asn:jacobian:i} For each $k\ge 0$, $J_k(\theta)$ is continuously differentiable on $\Theta$.
\item \label{asn:jacobian:ii} The SSJ and its derivative are uniformly summable on $\Theta$, i.e.,
\[
\sum_{k=0}^\infty \sup_{\theta\in\Theta}\|J_k(\theta)\| < \infty,\quad \sum_{k=0}^\infty \sup_{\theta\in\Theta}\|\partial \ve(J_k(\theta))/\partial \theta'\| < \infty.
\]
\end{asnitem}
\end{asn}

This assumption requires the SSJs $\lbrace J_k(\theta) \rbrace_k$ to be sufficiently smooth and asymptote to zero sufficiently fast at long horizons $k \to \infty$. In other words, agents' \emph{aggregated} optimal choices must depend smoothly on the structural parameters, and the dependence of current choices on expectations of macro sufficient statistics far into the future must be limited. Note that we do not impose that \emph{individual} agent choices are smooth, nor do we require a finite decision horizon.

% ============================================================
\subsection{Identification}\label{sec:identification}

We now give high-level assumptions for identification of the structural parameters and subsequently discuss their economic interpretation.

Recall that \cref{asn:moment_cond} assumed that $\cov(y_t,z_t) = \sum_{k=0}^\infty J_k(\theta) \cov(x_{t+k},z_t)$ at the true parameters $\theta=\theta_0$. The following assumption implies that the moment conditions are satisfied \emph{only} at the true parameters.

\begin{asn}[Identification]
\label{asn:identification}
The map $\theta \mapsto \sum_{k=0}^\infty J_k(\theta)\cov(x_{t+k},z_t)$ is injective on $\Theta$.\footnote{The infinite series is well-defined for all $\theta$ by \cref{asn:stationarity,asn:jacobian}.}
\end{asn}

Intuitively, this assumption requires the instrument vector $z_t$ to correlate with the macro sufficient statistics $x_{t+k}$ at a sufficient number of future horizons $k$ in order to distinguish the model-implied optimal agent choices $\sum_{k=0}^\infty J_k(\theta) \cov(x_{t+k},z_t)$ at different parameter values $\theta$. In particular, if the macro sufficient statistics were uncorrelated with all leads and lags of the instruments, identification would clearly fail. As a more interesting example, if a specific element in the parameter vector only governs the degree to which agents' behavior is forward-looking (that is, $J_k(\theta)$ is constant in $\theta$ at short horizons), then the instruments must predict the macro sufficient statistics at sufficiently long horizons; otherwise the exogenous variation induced by the instruments would not be helpful for identifying this specific parameter.

\begin{exm} \label{exm:linear}
Consider a stylized economic model where agents' aggregated consumption $y_t$ (a scalar) depends on expectations of current and next-period aggregate income $x_t$ (also a scalar), as well as a function $\xi_{t-1}$ of lagged shocks:
\[y_t = \beta x_t + \gamma E_t[x_{t+1}] + \xi_{t-1}.\]
The two unknown parameters are $\theta=(\beta,\gamma)'$. This corresponds to the SSJs $J_0(\theta) = \beta$, $J_1(\theta) = \gamma$, and $J_k(\theta) = 0$ for $k \geq 2$. Let $z_t$ be a vector of observed shocks. For this linear model, \cref{asn:identification} reduces to the requirement that the $2 \times d_z$ matrix
\[\cov\left(\begin{pmatrix}
x_t \\
x_{t+1}
\end{pmatrix}, z_t \right)\]
have full row rank $2$. In particular, there must be at least two instruments, and they must be able to ``perturb'' the future value $x_{t+1}$ of income differently from the current value $x_t$, in order to distinguish the contemporaneous coefficient $\beta$ from the forward-looking coefficient $\gamma$. That is, the impulse response functions of $x_t$ with respect to the various shocks in $z_t$ must not only be nonzero but also differ across shocks.

Now suppose instead that our structural model imposes the restriction $\beta=\gamma$, so there is only a single unknown parameter. This is a stylized version of a permanent income model. Then \cref{asn:identification} reduces to the weaker requirement that $\cov(x_t+x_{t+1},z_t) \neq 0$. So a single instrument suffices, and it just needs to induce variation in the \emph{cumulated} response of income. \qed
\end{exm}

For the asymptotic normality result below, we additionally require local identification.

\begin{asn}[Local identification]
\label{asn:invertibility}
The $d_g \times d_\theta$ matrix \[G_0
\equiv
-\sum_{k=0}^\infty \left(\cov(z_t,x_{t+k})\otimes I_{d_y}\right)
\frac{\partial \ve(J_k(\theta_0))}{\partial \theta'}\]
has full column rank.
\end{asn}

This assumption would imply \cref{asn:identification} if the SSJs were linear in the parameters $\theta$ (as in \cref{exm:linear}). In the nonlinear case, \cref{asn:invertibility} rules out pathologies where the map from parameters to outcomes is locally flat around the truth.

% ============================================================
\subsection{Consistency and asymptotic normality}\label{sec:consistency_AN}
We now establish the consistency and asymptotic normality of the minimum distance estimator $\widehat{\theta}$ defined in \cref{sec:estimator}. The key step in our proofs is to show that \cref{asn:stationarity,asn:jacobian} together imply pointwise convergence and stochastic equicontinuity of the sample moments $\widehat{g}(\theta)$. The pointwise convergence is a standard result \citep[e.g.,][Chapter 7.6]{Brillinger1981}. Stochastic equicontinuity follows from \citet{Dahlhaus1988} under strong restrictions on moments of the data of all orders. However, due to the \emph{bilinearity} of $\widehat{g}(\theta)$ in the periodogram $\widehat{S}(\cdot)$ and the SSJs $J(\cdot;\theta)$, we are able to establish stochastic equicontinuity under the same fourth-moment conditions in \cref{asn:stationarity} that are also required for pointwise convergence.

We impose the following conventional assumption on the weight matrix $\widehat{W}$.
\begin{asn}[Weight matrix]
\label{asn:weight-convergence}
$\widehat W\pto W$, where $W$ is symmetric positive definite.
\end{asn}

Consistency then follows from the standard results for extremum estimators \citep[Theorem 2.1]{NeweyMcFadden1994}.

\begin{prop}[Consistency]
\label{prop:consistency}
Under \cref{asn:moment_cond,asn:parameter-space,asn:stationarity,asn:jacobian,asn:weight-convergence,asn:identification}, \[\widehat\theta\pto \theta_0.\]
\end{prop}

To prove asymptotic normality of the estimator $\widehat{\theta}$, we impose a high-level central limit theorem (CLT) for the \emph{spectral mean} $\widehat{g}(\theta_0) = \frac{2\pi}{T} \sum_{j=0}^{T-1} \ve\lbrace \widehat{S}_{yz}(\omega_j)-J(\omega_j;\theta_0)\widehat{S}_{xz}(\omega_j) \rbrace$.

\begin{asn}[Spectral mean CLT]
\label{asn:clt}
Assume that \[
\sqrt{T}\widehat g(\theta_0) \dto N(0,\Omega)
\]
for some positive semidefinite $\Omega$.
\end{asn}

Recall that the population analogue of $\widehat{g}(\theta_0)$ equals 0 by \cref{asn:moment_cond}. Various sets of sufficient conditions for \cref{asn:clt} are provided by \citet[Theorem 7.6.1]{Brillinger1981}, \citet{Dahlhaus1988}, and \citet{MeyerPaparoditis2023}. Moreover, the spectral mean CLT follows from conventional time-domain CLTs under some further assumptions; see \cref{rem:time-domain-CLT} in \cref{app:proofs} for details.

Asymptotic normality now follows from the standard result for extremum estimators \citep[Theorem 3.2]{NeweyMcFadden1994}.

\begin{prop}[Asymptotic normality]
\label{prop:AN}
Under \cref{asn:moment_cond,asn:parameter-space,asn:stationarity,asn:jacobian,asn:weight-convergence,asn:identification,asn:clt,asn:invertibility}, and if $\theta_0$ lies in the interior of $\Theta$,
\[
\sqrt{T}(\widehat\theta-\theta_0)\dto
N\left(0,\,(G_0'WG_0)^{-1}G_0'W\Omega WG_0(G_0'WG_0)^{-1}\right).
\]
\end{prop}

The asymptotic null distribution of the over-identification statistic $\widehat{\Upsilon}$ in \cref{sec:overid} follows immediately from the calculations in \citet[Section 9.5]{NeweyMcFadden1994}.

\begin{cor}[Over-identification test]
\label{cor:overid}
Under the assumptions of \cref{prop:AN}, and if moreover $\widehat{\Omega} \pto \Omega$ and $\Omega$ is non-singular, then $\widehat{\Upsilon} \dto \chi_{d_g-d_\theta}^2$.
\end{cor}

\cref{app:hac} establishes the consistency of the HAC estimator of the variance-covariance matrix $\widehat{\Omega}$, defined in \cref{sec:inference}.

\section{Simulations}
\label{sec:simul}

We illustrate the empirical utility of our econometric procedures by applying them to data simulated from the workhorse two-asset HANK model in \citetalias{AuclertBardoczyRognlieStraub2021}, which builds on \citet{KaplanMollViolante2018}.

\paragraph{Structural model.}
The household block of the model is the one defined in \cref{sec:model}, and the general equilibrium is completed by specifying assumptions for financial intermediaries, firms, labor unions, monetary and fiscal policy, and market clearing. We refer to \citetalias{AuclertBardoczyRognlieStraub2021} (Appendix B.3) for details.

The true structural parameters are selected exactly as in \citetalias{AuclertBardoczyRognlieStraub2021} (Table B.III). In particular, the households' idiosyncratic productivity process $e_{i,t}$ evolves as an AR(1) process in logs, and the adjustment cost function for illiquid assets has the functional form
\[\Psi\left(a_{i,t},(1+r_t^a)a_{i,t-1}\right) \equiv \frac{\chi_1}{2}\frac{\left(a_{i,t}-(1+r_t^a)a_{i,t-1}\right)^2}{(1+r_t^a)a_{i,t-1}+\chi_0},\]
where $\chi_0,\chi_1>0$ are parameters.\footnote{In other words, we maintain the same calibration $\chi_2=2$ as \citetalias{AuclertBardoczyRognlieStraub2021}, using their notation.}

We assume the macroeconomy is driven by a subset of the shocks in the empirically estimated two-asset HANK model of \citetalias{AuclertBardoczyRognlieStraub2021} (Section 5.4): a TFP shock, a monetary policy shock, and a government spending shock. We limit ourselves to this list of shocks for purely pedagogical reasons, since this allows us to use the two-asset model definition in the \texttt{sequence-jacobian} Python package entirely off the shelf; however, it would be easy to add additional shocks (e.g., to markups). The shocks are assumed to follow independent Gaussian AR(1) processes with parameters given by the posterior means reported by \citetalias{AuclertBardoczyRognlieStraub2021} (Table F.III, columns labeled ``Posterior (Shocks)'').

\paragraph{Simulation and estimation settings.}
We assume that the researcher observes $T=120$ quarterly observations simulated from the linearized equilibrium of the full structural model.\footnote{The choice of sample size is motivated by the available time spans of cross-sectional data on consumption and savings in data sources like the Consumer Expenditure Survey and Distributional Financial Accounts.} The econometrician exploits $d_y=9$ household output series:
\begin{itemize}
\item The cross-sectional mean, variance, and centered third moment of log consumption.
\item The total amount of liquid assets held by households (i) below the median of the wealth distribution, (ii) between the 50th--90th percentiles of the wealth distribution, and (iii) between the 90th--99th percentiles of the wealth distribution.\footnote{Here ``wealth distribution'' refers to liquid plus illiquid assets.}
\item Same statistics as in the previous bullet point, but for illiquid assets.
\end{itemize}
For added realism, and to avoid numerical issues arising from dynamically singular moments, we add temporally and cross-sectionally independent Gaussian measurement errors to each of the 9 observed output series. For each series, the standard deviation of the measurement error equals 20\% of the population standard deviation of the original un-contaminated series.

The econometrician additionally observes the $d_x=3$ macro sufficient statistics: returns on liquid and illiquid assets, and the aggregate post-tax real earnings. Finally, they observe $d_z=2$ instruments, namely the time-$t$ innovations to the government spending and monetary policy shocks.

We estimate four parameters: the household EIS $\gamma$ and discount factor $\beta$, and the illiquid asset adjustment cost parameters $\chi_0$ and $\chi_1$. The researcher is assumed to have correct knowledge of the true values of the remaining parameters in the household block.\footnote{These include the idiosyncratic productivity process parameters and the steady-state values of the macro sufficient statistics. In practice, the latter values could be estimated from long-run time series averages.} The estimation procedure only requires knowledge of the structure of the household block, and does not use any information about the rest of the general equilibrium (including the nature of any unobserved variables or shocks). We use the \emph{ad hoc} minimum distance weight matrix proposed in \cref{sec:weight_mat}, and compute the estimator using numerical optimization.\footnote{We use the gradient-based L-BFGS-B algorithm implemented in SciPy's \texttt{minimize} routine. In each simulation, we first draw 5 initial parameter values at random, perform 5 optimization steps for each, and record the resulting parameter vector that yields the lowest objective function value across the 5 attempts; then the final, full optimization is initialized at this single parameter vector. The 5 random draws of initial values are uniform on $[0.5 \times \text{truth},1.5 \times \text{truth}]$ (independently across the four parameters), except for $\beta$ which due to its natural parameter space is drawn from $[0.99 \times \text{truth},1.01 \times \text{truth}]$. The optimization routine enforces some very loose bounds on the parameters that never bind at the computed optimum in any simulation.}

In addition to analytical confidence intervals and over-identification tests, we consider the bootstrap procedures developed in \cref{app:bootstrap} (without the correction for potentially non-Gaussian shocks). Due to the high computational cost of repeatedly running nonlinear optimizations, we approximate the performance of bootstrap procedures using the warp-speed diagnostic of \citet{GiacominiPolitisWhite2013}.\footnote{That is, we only generate a single bootstrap draw in each simulated data set. The bootstrap quantiles for the confidence interval and over-identification test are computed from the distribution of bootstrap draws across simulated data sets, rather than across bootstrap iterations in a given simulated data set. This strategy is of course infeasible in reality, where only a single data set is available. The warp-speed diagnostic checks whether the bootstrap distribution of a statistic lines up with its sampling distribution, after integrating out over the distribution of the data. As such, the diagnostic can detect various types of bootstrap failure, but not all \citep[p.\ 575]{GiacominiPolitisWhite2013}.}

\paragraph{Results.}

\begin{table}[t]
\centering
\begin{tabular}{@{}lrrrrrr@{}} \toprule
\multicolumn{5}{c}{ } & Analyt.\ & Bootstr.\  \\
\cmidrule(lr){6-7}
Param.\ & Truth & Bias & Stdev. & RMSE & \multicolumn{2}{c}{CI coverage} \\
\midrule
$\gamma$ & 0.5000 & 0.0214 & 0.0844 & 0.0871 & 0.467 & 0.904 \\
$\beta$ & 0.9763 & 0.0000 & 0.0032 & 0.0032 & 0.479 & 0.912 \\
$\chi_0$ & 0.2500 & 0.0374 & 0.0893 & 0.0969 & 0.990 & 0.947 \\
$\chi_1$ & 6.4164 & -0.7883 & 1.7271 & 1.8985 & 0.621 & 0.857 \\
\midrule
\multicolumn{5}{c}{ } & \multicolumn{2}{c}{Rejection rate} \\
\cmidrule(lr){6-7}
\multicolumn{5}{@{}l}{Over-ID test} & 0.949 & 0.115 \\
\bottomrule
\end{tabular}
\caption{Simulation results across 512 Monte Carlo replications. Sample size: $T=120$. Nominal significance level: $\alpha=0.1$. Columns left to right: parameter name, true value, estimator bias, estimator standard deviation, estimator root mean squared error, coverage of analytical confidence interval, and coverage of bootstrap confidence interval. The last row reports the rejection rate of the over-identification test, either with analytical (left) or bootstrap (right) critical value. Bootstrap coverage and rejection rates are approximated with the warp-speed diagnostic of \citet{GiacominiPolitisWhite2013}.} \label{tab:simul}
\end{table}

\cref{tab:simul} shows that our estimator is able to deliver accurate inference on the four parameters in the household block despite the moderate sample size. The finite-sample biases and standard deviations of all parameter estimators are modest relative to the scale of the true parameter values. As expected due to the limited sample and the nonlinearity of the SSJs in the underlying parameters, the analytical confidence interval developed in \cref{sec:inference} under-covers substantially for 3 of the 4 parameters, though the extent of under-coverage may be tolerable for initial experimentation. However, the bootstrap confidence interval described in \cref{app:bootstrap} achieves near-nominal coverage for all parameters. As for the over-identification test, the bootstrap version has near-nominal rejection rate, while the analytical critical value leads to severe over-rejection.

We conclude that the parameter estimators, as well as the bootstrap confidence intervals and tests, perform well in an empirically calibrated DGP with modest sample size $T=120$. By contrast, analytical confidence intervals should be used only as a rough guide, and the analytical critical value for the over-identification test is not reliable in our simulation DGP.

%\section{Application}
%\label{sec:appl}
%
%[TO DO XX]

\section{Conclusion}
\label{sec:concl}

We proposed a general procedure for estimating and testing a single block of a heterogeneous agent model. Like the existing Generalized Method of Moments procedures used to estimate individual model blocks in representative agent settings, our limited-information method is fully consistent with general equilibrium but does not require imposing assumptions on the remaining blocks of the model. Our procedure is conceptually and computationally easy to implement on any heterogeneous agent model block that fits into the sequence-space framework of \citetalias{AuclertBardoczyRognlieStraub2021}.

The utility of estimating a single model block is threefold. First, the parameters of this block may be of direct scientific interest, or they can serve as robust calibration inputs into a subsequent fully-specified structural model. Second, one can estimate any policy counterfactuals that are functions only of the estimated model block. Third, we developed an over-identification test that can evaluate the empirical validity of the model block in isolation, without the risk of contamination from other potentially misspecified model blocks.

\clearpage
\appendix
\crefalias{section}{appsec}

% Appendix: number theorem-type environments per appendix section (A.1, B.1, C.1, ...).
% Counters stay globally monotonic (never reset), so cleveref orders main-text results
% before appendix ones; per-section restart of the printed number is done via a base offset.
\newcounter{asnbase}\newcounter{propbase}
\setcounter{asnbase}{\value{asn}}
\setcounter{propbase}{\value{prop}}
% restart the appendix count of assumptions/propositions at each appendix section
\AddToHook{cmd/section/before}{\setcounter{asnbase}{\value{asn}}\setcounter{propbase}{\value{prop}}}

\makeatletter
% asn and prop also occur in the main text, so subtract a base to start the appendix count at 1
\renewcommand{\theasn}{\thesection.\number\numexpr\c@asn-\c@asnbase\relax}
\renewcommand{\theprop}{\thesection.\number\numexpr\c@prop-\c@propbase\relax}
% keep hyperref anchors unique even though the printed form restarts each section
\@ifundefined{theHasn}{}{\renewcommand{\theHasn}{\number\c@asn}}
\@ifundefined{theHprop}{}{\renewcommand{\theHprop}{\number\c@prop}}
\makeatother

\renewcommand{\thelem}{\thesection.\arabic{lem}}
\renewcommand{\thethm}{\thesection.\arabic{thm}}
\renewcommand{\thecor}{\thesection.\arabic{cor}}
\renewcommand{\theconj}{\thesection.\arabic{conj}}
\renewcommand{\thedefn}{\thesection.\arabic{defn}}
\renewcommand{\theexm}{\thesection.\arabic{exm}}
\renewcommand{\theremark}{\thesection.\arabic{remark}}

\begin{appendices}
\section{Bootstrap inference}
\label{app:bootstrap}

Here we describe a bootstrap inference approach that is designed to capture the nonlinearity of the SSJs in the structural parameters.

\paragraph{Multiplier bootstrap.}
The bootstrap procedures build on the Gaussian multiplier bootstrap. This procedure assumes that the data is Gaussian \citep{MeyerPaparoditis2023}. We discuss how to correct for non-Gaussian data below.

The bootstrap requires an estimate $\widehat{f}(\omega)$ of the spectral density of the data $\zeta_t \equiv (y_t',x_t',z_t')'$. Letting $\widehat{S}(\omega)$ denote the periodogram of $\zeta_t$, we choose the periodogram smoothing estimator
\[\widehat{f}(\omega) \equiv \frac{\sum_{j=1}^{T-1} \widetilde{K}\left(\frac{\min\lbrace |\omega-\omega_j|,2\pi-|\omega-\omega_j|\rbrace}{B}\right)\widehat{S}(\omega_j)}{\sum_{j=1}^{T-1} \widetilde{K}\left(\frac{\min\lbrace |\omega-\omega_j|,2\pi-|\omega-\omega_j|\rbrace}{B}\right)}\]
with Epanechnikov kernel $\widetilde{K}(u)=\max\lbrace 1-u^2,0\rbrace$ and spectral bandwidth $B = T^{-0.2}$.\footnote{This bandwidth rule is approximately integrated-MSE-optimal for an AR(1) process with $\rho=0.7$, as can be verified by plugging into equations 6.2.108 and 6.2.119 in \citet{Priestley1981}. Details are available upon request.}

In each bootstrap iteration, the algorithm proceeds as follows:
\begin{enumerate}
	\item Draw pseudo-Fourier-transforms $F_j^\dagger \sim N_c(0, \widehat{f}(\omega_j), 0)$ independently across Fourier frequencies $\omega_j=2\pi j/T$, $j=1,\dots,\lfloor (T-1)/2 \rfloor$ (see \citealp[Appendix A.1]{MeyerPaparoditis2023}, for the definition of the complex normal distribution). Then let the pseudo-periodogram draw be given by $\widehat{S}^\dagger(\omega_j) \equiv F_j^\dagger (F_j^\dagger)^*$ for all $j=1,\dots,\lfloor (T-1)/2 \rfloor$ and $\widehat{S}^\dagger(\omega_j) \equiv \overline{\widehat{S}^\dagger(\omega_{T-j})}$ for the remaining $j\leq T-1$, where an asterisk denotes complex conjugate transpose and bar denotes elementwise complex conjugate. At the boundaries we set $\widehat{S}^\dagger(0)=0$ and (if $T$ is even) $\widehat{S}^\dagger(\pi)=0$.
	\item Compute the resampled estimator
	\[\widehat{\theta}^\dagger \equiv \argmin_{\theta \in \Theta}\; \widehat{g}^\dagger(\theta)'\widehat{W}\widehat{g}^\dagger(\theta),\]
	using the re-centered moment function\footnote{The re-centering ensures that the sample moment function has bootstrap expectation approximately equal to zero at the real-data estimate $\widehat{\theta}$.}
	\[\widehat{g}^\dagger(\theta) \equiv \frac{2\pi}{T}\sum_{j=0}^{T-1} \ve\left\lbrace \widehat{S}_{yz}^\dagger(\omega_j) - J(\omega_j;\theta) \widehat{S}_{xz}^\dagger(\omega_j) - \left[\widehat{f}_{yz}(\omega_j) - J(\omega_j;\widehat{\theta}) \widehat{f}_{xz}(\omega_j)\right] \right\rbrace,\]
	where the ``$yz$'' and ``$xz$'' subscripts indicate the $(y_t,z_t)$ and $(x_t,z_t)$ blocks of the cross-periodogram or cross-spectrum.\footnote{While $\widehat{g}^\dagger(\theta)$ should be real in large samples, in finite samples we drop any imaginary part.}
\end{enumerate}
Though we do not formally prove bootstrap validity, the analysis in \citet{MeyerPaparoditis2023} suggests that the resampling distribution of $\widehat{\theta}^\dagger-\widehat{\theta}$ (conditional on the data) approximates the (unconditional) sampling distribution of $\widehat{\theta}-\theta_0$, provided the data is Gaussian. The bootstrap standard errors are the standard deviations of the elements of $\widehat{\theta}^\dagger$ across resamples. An (asymmetric) $1-\alpha$ confidence interval for $\theta_0$ is given by $[2\widehat{\theta} - q_{1-\alpha/2}^\dagger, 2\widehat{\theta} - q_{\alpha/2}^\dagger]$, where $q_\alpha^\dagger$ is the $\alpha$-th quantile of $\widehat{\theta}^\dagger$ across bootstrap samples.

\paragraph{Correction for non-Gaussian data.}
To allow for non-Gaussian shocks, we can rescale the multiplier bootstrap draws $\widehat{\theta}^\dagger$ using the Multivariate Frequency-domain Hybrid Bootstrap (MFHB) of \citet{MeyerPaparoditis2023}. The calculations in the proof of \cref{thm:AN} show that $\widehat{\theta}-\theta_0$ is asymptotically equivalent with $-(\widehat{G}'\widehat{W}\widehat{G})^{-1}\widehat{G}'\widehat{W}\widehat{g}(\theta_0)$, where we recall that $\widehat{g}(\theta_0)=\frac{2\pi}{T}\sum_{j=0}^{T-1} \ve\left\lbrace \widehat{S}_{yz}(\omega_j) - J(\omega_j;\theta_0) \widehat{S}_{xz}(\omega_j) \right\rbrace$ is a spectral mean. We can approximate the Gaussian and non-Gaussian variance-covariance matrices of $\widehat{g}(\theta_0)$ by applying the MFHB algorithm with $\theta_0 \approx \widehat\theta$ (fixed across bootstrap iterations); denote these variance-covariance matrices by $\widehat{\Omega}_\text{G}$ and $\widehat{\Omega}_\text{NG}$, respectively.\footnote{In the notation of \citet[p.\ 2379, see also Remark 3.5]{MeyerPaparoditis2023}, these are $\mathbf{G}_n^*$ and $\mathbf{G}_n^\circ$, respectively. Note that the proofs in \citet{MeyerPaparoditis2023} do not technically allow for data-dependent spectral means (which ours is due to the dependence on $\widehat\theta$), but we conjecture that the result can be extended to cover our setting.} Starting with the multiplier bootstrap draws $\widehat{\theta}^\dagger$, the rescaled draws
\begin{align*}
\widehat{\theta}_\text{rescale}^\dagger &\equiv \widehat{\theta} + \left\lbrace (\widehat{G}'\widehat{W}\widehat{G})^{-1}\widehat{G}'\widehat{W}\widehat{\Omega}_\text{NG}\widehat{W}\widehat{G}(\widehat{G}'\widehat{W}\widehat{G})^{-1}\right\rbrace^{1/2} \\
&\qquad\qquad \times \left\lbrace (\widehat{G}'\widehat{W}\widehat{G})^{-1}\widehat{G}'\widehat{W}\widehat{\Omega}_\text{G}\widehat{W}\widehat{G}(\widehat{G}'\widehat{W}\widehat{G})^{-1}\right\rbrace^{-1/2}(\widehat{\theta}^\dagger-\widehat{\theta})
\end{align*}
then have the correct asymptotic variance-covariance matrix (as in \cref{prop:AN}) even under non-Gaussianity.\footnote{We use the symmetric positive semidefinite matrix square root.} The bootstrap standard errors and confidence intervals can now be computed from the rescaled draws. %The rescaling step is available as an option in our code suite on GitHub.

\paragraph{Over-identification test.}
The bootstrap critical value for the over-identification test is obtained by recomputing the estimator and over-identification statistic from \cref{sec:overid} in every bootstrap data set. As above, we use the re-centered moments $\widehat{g}^\dagger(\widehat{\theta}^\dagger)$ when computing the statistic, as this imposes the null hypothesis of correct specification on the bootstrap DGP. The only remaining challenge is that the HAC estimator $\widehat{\Omega}$ in the formula for the statistic is defined for data in the time domain, but we bootstrap the periodogram rather than the underlying data. Hence, for every draw of the pseudo-Fourier-transforms $F_j^\dagger$, we apply the (scaled) inverse Fourier transform to compute the corresponding data $\lbrace \zeta_t^\dagger \rbrace$ in the time domain. Then we apply the HAC procedure to this data (just as with the real data).

Being based on Gaussian draws, this procedure is designed to work well when the real data is Gaussian, but we conjecture it remains \emph{asymptotically} valid even if the real data is non-Gaussian. This is because as $T \to \infty$, the bootstrap distribution of the over-identification statistic converges to a $\chi^2$ distribution with $d_g-d_\theta$ degrees of freedom, since the statistic is asymptotically pivotal. This is the correct distribution for computing asymptotic critical values whether or not the data is Gaussian, cf.\ \cref{cor:overid}.

\section{Asymptotic variance estimation}
\label{app:hac}
Here we state the consistency of the HAC long-run variance estimator. Recall its definition from \cref{sec:inference}:
\[
\widehat\Omega
\equiv \sum_{|h|\le L_T} K\left(\frac{h}{L_T}\right)\widehat\Gamma_{\psi}(h),\quad \text{where} \quad \widehat\Gamma_{\psi}(h)\equiv \frac{1}{T}\sum_{t=1}^{T}\widehat\psi_t\widehat\psi_{t-h}',\quad \widehat{\psi}_t \equiv \psi_{t,T}(\widehat\theta),
\]
and we define the truncated triangular array score
\[
\psi_{t,T}(\theta)
\equiv
\ve\left(\left(y_t-\sum_{k=0}^{T-t}J_k(\theta)x_{t+k}\right)z_t'\right)\1\{1\le t\le T\}.
\]
For simplicity, we here impose in the definition of the HAC estimator that the data has mean zero (\cref{asn:stationarity}), and we do not de-mean the sample score $\widehat\psi_t$ itself. We conjecture that both types of sample de-meaning have negligible effects asymptotically under our assumptions (see \cref{sec:est-var}).

Define the population score process
\[
\psi_t(\theta) \equiv \ve\left(\left(y_t-\sum_{k=0}^{\infty}J_k(\theta)x_{t+k}\right)z_t'\right).
\]
Then $g(\theta) \equiv \int_0^{2\pi} \ve\lbrace S_{yz}(\omega)-J(\omega;\theta)S_{xz}(\omega)\rbrace\,d\omega = \E[\psi_t(\theta)]$.
Define the long-run variance
\[
\Omega \equiv \sum_{h\in\mathbb Z}\Gamma_\psi(h),\quad \text{where}\quad \Gamma_\psi(h) \equiv \E\left[\psi_h(\theta_0)\psi_0(\theta_0)'\right].
\]
This is the same $\Omega$ that appears in \cref{asn:clt} whenever that assumption (and our other assumptions) hold; see \cref{rem:time-domain-CLT} below.

\begin{asn}[HAC kernel and bandwidth]
\label{asn:HAC-kernel}
Assume:
\begin{asnitem}
\item \label{asn:HAC-kernel:i} The kernel $K(\cdot)$ is symmetric, bounded, supported on $[-1,1]$, and continuous at $0$ with $K(0)=1$.
\item \label{asn:HAC-kernel:ii} The bandwidth satisfies $L_T\to\infty$ and $L_T/\sqrt{T}\to 0$.
\end{asnitem}
\end{asn}
This assumption follows \citet{Andrews1991} and is satisfied by standard kernels, such as Newey-West (Bartlett).

\begin{asn}[Stronger cumulant summability]
\label{asn:stronger-stationarity}
Assume $\zeta_t$ has finite 8th moments, and strengthen the cumulant summability condition in \cref{asn:stationarity:ii} to hold for all $\ell=2,3,\dots,8$.
\end{asn}
Eighth-order cumulant control on $\zeta_t$ is used to control fourth-order cumulants of the score,
which is quadratic in $\zeta_t$.

\begin{asn}[Weighted summability of SSJ]
\label{asn:jacobian-tail-summability}
\(
\sum_{k=0}^\infty k\|J_k(\theta_0)\| < \infty.
\)
\end{asn}

The weighted summability condition ensures that the sample boundary truncation in the triangular array score $\psi_{t,T}(\theta_0)$ is asymptotically negligible.

The HAC estimation error decomposes into four terms that each vanish: error from plugging in $\widehat\theta$ for $\theta_0$, stochastic fluctuation of the covariance average, kernel bias, and truncation of the infinite lag sum. Therefore, the HAC estimator is consistent.

\begin{prop}[Consistent HAC estimator]
\label{prop:HAC-Omega}
Under the assumptions of \cref{prop:AN}, and if moreover \cref{asn:HAC-kernel,asn:stronger-stationarity,asn:jacobian-tail-summability} hold, then
\[
\widehat\Omega
\pto \Omega.
\]
\end{prop}

\section{Proofs}\label{app:proofs}

In this appendix, we develop asymptotic theory for a general ``spectral GMM'' estimator, and then specialize it to the setup in \cref{sec:framework,sec:theory}. \cref{sec:general-spectral-gmm} develops consistency, asymptotic normality, and HAC variance estimation in the general setup, \cref{sec:aux-general-spectral-gmm} collects auxiliary results, and \cref{sec:specialization-our-setup} specializes to our applications.

Let the data $\zeta_t\in\mathbb R^{d_\zeta}$ be real-valued, with spectral density matrix $S(\omega)\in\mathbb C^{d_\zeta\times d_\zeta}$ for $\omega\in[0,2\pi]$. Define the discrete Fourier transform and the matrix periodogram
\[
F_\zeta(\omega)\equiv \frac{1}{\sqrt{T}}\sum_{t=1}^T \zeta_t e^{-\iota\omega t},
\quad
\widehat S(\omega)\equiv \frac{1}{2\pi}F_\zeta(\omega)F_\zeta(\omega)^*,
\]
where $A^*$ denotes the conjugate transpose of $A$ for a generic vector/matrix $A$. 
Let $\widetilde J(\omega;\theta)\in\mathbb C^{d_g\times d_\zeta^2}$ be a function that is measurable in $\omega$ and
indexed by $\theta\in\Theta\subset\mathbb R^{d_\theta}$. Define the population and sample spectral means
\[
g(\theta)\equiv \int_0^{2\pi}\widetilde J(\omega;\theta)\ve(S(\omega))\,d\omega,
\quad
\widehat g(\theta)\equiv \frac{2\pi}{T}\sum_{j=0}^{T-1}\widetilde J(\omega_j;\theta)\ve(\widehat S(\omega_j)).
\]
We refer to $\widetilde{J}$ as a spectral weight function, and restrict attention to the following class.
\begin{asn}[Conjugate symmetry of spectral weight function]
\label{asn:J-conjugate-symmetry} For any $\theta\in\Theta$, $\widetilde{J}(0;\theta)$ is real, and
\[
\widetilde J(2\pi-\omega;\theta) = \overline{\widetilde J(\omega;\theta)}\quad \text{for all } \omega\in[0,2\pi].
\]
Here $\overline{A}$ denotes the elementwise conjugate of $A$ for a generic matrix $A$.
\end{asn}
\cref{asn:J-conjugate-symmetry} and $\zeta_t \in \mathbb{R}^{d_\zeta}$ imply that $g(\theta),\widehat g(\theta)\in\mathbb R^{d_g}$. We consider a spectral GMM estimator $\widehat\theta$ defined by
\[\widehat\theta\in\argmin_{\theta\in\Theta}\widehat Q(\theta),\quad \text{where} \quad \widehat Q(\theta) \equiv \widehat g(\theta)'\widehat W\widehat g(\theta).\]
Define also the corresponding population objective function $Q(\theta)\equiv g(\theta)'Wg(\theta)$.

For clarity, we now briefly indicate how this general framework specializes to the application in \cref{sec:framework,sec:theory}. Let $M_z\in\mathbb R^{d_z\times d_\zeta}$ be the selection matrix formed by the last $d_z$ rows of the identity matrix $I_{d_\zeta}$, so that $z_t=M_z\zeta_t$. Define the spectral weight function
\[
\widetilde J_\text{special}(\omega;\theta) \equiv M_z\otimes \left(I_{d_y},\,-J(\omega;\theta),\,0_{d_y\times d_z}\right)\in\mathbb C^{d_g\times d_\zeta^2}.
\]
Then
\[
\widetilde J_\text{special}(\omega;\theta)\ve(S(\omega))
=\ve\lbrace S_{yz}(\omega)-J(\omega;\theta)S_{xz}(\omega)\rbrace.
\]

\subsection{Asymptotic results for general spectral GMM}\label{sec:general-spectral-gmm}
% ============================================================
\subsubsection{Consistency}

\begin{asn}[Regularity of spectral weight function]
\label{asn:J-regularity}
Assume:
\begin{asnitem}
\item \label{asn:J-regularity:i} (Riemann integrability in $\omega$.) For each $\theta\in\Theta$,
$\omega\mapsto \widetilde J(\omega;\theta)$ is bounded and piecewise continuous on $[0,2\pi]$.
\item \label{asn:J-regularity:ii} (Uniform continuity in $\theta$.) Define the modulus
\[
\Delta_{\widetilde J}(\delta)\equiv \sup_{\theta,\widetilde\theta \in \Theta \colon \|\widetilde\theta-\theta\|\le\delta}\sup_{\omega\in[0,2\pi]}
\|\widetilde J(\omega;\widetilde\theta)-\widetilde J(\omega;\theta)\|.
\]
Then $\Delta_{\widetilde J}(\delta)\to 0$ as $\delta\to 0$.
\end{asnitem}
\end{asn}

\begin{asn}[Identification: general case]
\label{asn:identification-general}
For any $\theta \in \Theta$, $g(\theta)=0$ if and only if $\theta=\theta_0$.
\end{asn}

\begin{thm}[Consistency: general case]
\label{thm:consistency}
Under \cref{asn:parameter-space,asn:stationarity,asn:weight-convergence,asn:identification-general,asn:J-conjugate-symmetry,asn:J-regularity}, \[\widehat\theta\pto \theta_0.\]
\end{thm}
\begin{proof}
\cref{asn:parameter-space} gives compact $\Theta$. Under \cref{asn:parameter-space,asn:stationarity,asn:J-regularity}, \cref{rem:g-continuity} implies $g(\theta)$ is continuous, so $Q(\theta)$ is continuous. Under \cref{asn:weight-convergence,asn:identification-general}, $Q(\theta)$ has a unique minimizer at $\theta_0$ \citep[Lemma 2.3]{NeweyMcFadden1994}.

Under \cref{asn:parameter-space,asn:stationarity,asn:J-regularity}, \cref{lem:uniform-g} yields $\sup_{\theta\in\Theta}\|\widehat g(\theta)-g(\theta)\|\pto 0$. Since $g$ is continuous on compact $\Theta$, we have that $\sup_{\theta}\|g(\theta)\|<\infty$ and $\sup_{\theta}\|\widehat g(\theta)\|=O_p(1)$. Thus, under \cref{asn:weight-convergence},
\begin{align*}
\sup_{\theta\in\Theta}\left|\widehat Q(\theta)-Q(\theta)\right|
&\le \|\widehat W-W\|\sup_{\theta}\|\widehat g(\theta)\|^2
+\|W\|\sup_{\theta}\|\widehat g(\theta)-g(\theta)\|\sup_{\theta}\left(\|\widehat g(\theta)\|+\|g(\theta)\|\right)\\
&\pto 0.
\end{align*}
Consistency then follows from \citet[Theorem 2.1]{NeweyMcFadden1994}.
\end{proof}

% ============================================================
\subsubsection{Asymptotic normality}

\begin{asn}[Differentiability of spectral weight function]
\label{asn:J-deriv}
There exists a compact neighborhood $\Theta_0\subset\Theta$ of $\theta_0$ such that:
\begin{asnitem}
\item \label{asn:J-deriv:i} For each $\omega\in[0,2\pi]$, $\theta\mapsto \ve(\widetilde J(\omega;\theta))$ is differentiable on $\Theta_0$.

\item \label{asn:J-deriv:ii} (Riemann integrability of derivative in $\omega$.) For each $\theta\in\Theta_0$, $\omega\mapsto \partial \ve(\widetilde J(\omega;\theta))/\partial \theta'$ is bounded and piecewise continuous on $[0,2\pi]$.
\item \label{asn:J-deriv:iii} (Uniform continuity of derivative in $\theta$.)
\[
\Delta_{\partial\widetilde J}(\delta)
\equiv \sup_{\theta,\widetilde\theta\in\Theta_0\colon \|\widetilde\theta-\theta\|\le\delta}\sup_{\omega\in[0,2\pi]}
\left\|\frac{\partial \ve(\widetilde J(\omega;\widetilde\theta))}{\partial \theta'}-\frac{\partial \ve(\widetilde J(\omega;\theta))}{\partial \theta'}\right\|
\to 0
\quad\text{as }\delta\to 0.
\]
\end{asnitem}
\end{asn}

For $\theta\in\Theta_0$, we have
\begin{align*}
G(\theta)&\equiv \frac{\partial g(\theta)}{\partial \theta'} =\int_0^{2\pi}\lbrace \ve(S(\omega))' \otimes I_{d_g} \rbrace\frac{\partial \ve(\widetilde J(\omega;\theta))}{\partial \theta'}\,d\omega,\\
\widehat G(\theta)&\equiv \frac{\partial \widehat g(\theta)}{\partial \theta'}=\frac{2\pi}{T}\sum_{j=0}^{T-1}\lbrace \ve(\widehat S(\omega_j))' \otimes I_{d_g} \rbrace\frac{\partial \ve(\widetilde J(\omega_j;\theta))}{\partial \theta'}.
\end{align*}

\begin{asn}[Local identification: general case]
\label{asn:invertibility-general}
$G_0 \equiv G(\theta_0)$ has full column rank $d_\theta$.
\end{asn}

\begin{thm}[Asymptotic normality: general case]
\label{thm:AN}
Under \cref{asn:parameter-space,asn:stationarity,asn:weight-convergence,asn:identification-general,asn:clt,asn:invertibility-general,asn:J-conjugate-symmetry,asn:J-regularity,asn:J-deriv},
\[
\sqrt{T}(\widehat\theta-\theta_0)\dto
N\left(0,\,(G_0'WG_0)^{-1}G_0'W\Omega WG_0(G_0'WG_0)^{-1}\right).
\]
\end{thm}
\begin{proof}
Under \cref{asn:parameter-space,asn:stationarity,asn:weight-convergence,asn:identification-general,asn:J-conjugate-symmetry,asn:J-regularity}, \cref{thm:consistency} gives consistency. \cref{asn:J-deriv} implies $\theta_0\in\mathrm{int}(\Theta)$, \cref{asn:clt} gives the CLT for $\widehat{g}(\theta_0)$, and \cref{asn:weight-convergence,asn:invertibility-general} give invertibility of $G_0'WG_0$. \cref{asn:J-deriv:i,asn:J-deriv:iii} yield that $\widehat g(\theta)$ is continuously differentiable on $\Theta_0$. Under \cref{asn:stationarity,asn:J-deriv}, \cref{lem:uniform-G} yields uniform convergence of $\widehat G(\theta)$ to $G(\theta)$ on $\Theta_0$, and \cref{rem:G-continuity} yields continuity of $G(\theta)$ at $\theta_0$.
The conclusion then follows from \citet[Theorem 3.2]{NeweyMcFadden1994}.
\end{proof}

% ============================================================
\subsubsection{HAC}
\label{sec:est-var}

Let
\[\widetilde J_k(\theta)\equiv \frac{1}{2\pi}\int_0^{2\pi} e^{-\iota\omega k}\widetilde J(\omega;\theta)\,d\omega,\quad k\in\mathbb Z,\]
denote the time-domain coefficient
matrices associated with the spectral weight function $\widetilde J(\cdot;\theta)$. By \cref{asn:J-conjugate-symmetry}, each $\widetilde J_k(\theta)$ is real-valued.

The following assumption restricts the smoothness and tail behavior of the time-domain coefficients $\lbrace \widetilde{J}_k(\theta) \rbrace_k$. \cref{rem:Fourier-series-summability-sufficient} below argues that the assumption can be verified purely by checking smoothness conditions on the spectral weight function $\widetilde{J}(\omega;\theta)$, without explicitly computing the time-domain coefficients.

\begin{asn}[Fourier series summability]
\label{asn:Fourier-series-summability}
There exists a compact neighborhood $\Theta_0\subset\Theta$ of $\theta_0$ such that:
\begin{asnitem}
\item \label{asn:Fourier-series-summability:i} For each $k\in\mathbb Z$, $\widetilde J_k(\theta)$ is continuous on $\Theta$. $\widetilde J(\cdot;\theta)$ has an absolutely convergent Fourier series uniformly in $\Theta$, i.e.,
\[
\sum_{k\in\mathbb Z}\sup_{\theta\in\Theta}\|\widetilde J_k(\theta)\|<\infty,
\]
and $\widetilde{J}(\cdot;\theta)$ equals its Fourier series for all $\theta \in \Theta$.
\item \label{asn:Fourier-series-summability:ii} For each $k\in\mathbb Z$, $\widetilde J_k(\theta)$ is continuously differentiable on $\Theta_0$. $\partial\ve(\widetilde J(\cdot;\theta))/\partial\theta'$ has an absolutely convergent Fourier series uniformly in $\Theta_0$, i.e.,
\[
\sum_{k\in\mathbb Z}\sup_{\theta\in\Theta_0}\|\partial\ve(\widetilde J_k(\theta))/\partial\theta'\|<\infty.
\]
\item \label{asn:Fourier-series-summability:iii} $\partial_\omega\widetilde J(\cdot;\theta_0)$ has an absolutely convergent Fourier series, i.e.,\footnote{Note that
$(\partial_\omega\widetilde J)_k(\theta_0)
\equiv\frac{1}{2\pi}\int_0^{2\pi}e^{-\iota\omega k}\partial_\omega\widetilde J(\omega;\theta_0)\,d\omega
=\frac{\iota k}{2\pi}\int_0^{2\pi}e^{-\iota\omega k}\widetilde J(\omega;\theta_0)\,d\omega
=\iota k\widetilde J_k(\theta_0)$.}
\[
\sum_{k\in\mathbb Z}|k|\|\widetilde J_k(\theta_0)\|<\infty.
\]
\end{asnitem}
\end{asn}

Define the score process and its truncated triangular array counterpart:
\begin{align*}
\psi_t(\theta)&\equiv \sum_{k\in\mathbb Z}\widetilde J_k(\theta)\ve(\zeta_{t+k}\zeta_t'), \\
\psi_{t,T}(\theta) &\equiv
\sum_{k=-(t-1)}^{T-t}\widetilde J_k(\theta)\ve(\zeta_{t+k}\zeta_t')\1\{1\le t\le T\}.
\end{align*}
The latter uses all available data $\zeta_1,\dots,\zeta_T$.

Define the population long-run variance
\[\Omega\equiv \sum_{h\in\mathbb Z}\Gamma_\psi(h) \quad \text{where} \quad \Gamma_\psi(h)\equiv \E\left[\psi_h(\theta_0)\psi_0(\theta_0)'\right].\]
We estimate this long-run variance with the kernel HAC estimator
\[\widehat\Omega
\equiv \sum_{|h|\le L_T} K(h/L_T)\widehat\Gamma_{\psi,T}(h),\quad \text{where}\quad \widehat\Gamma_{\psi,T}(h)\equiv \frac{1}{T}\sum_{t=1}^{T}\psi_{t,T}(\widehat\theta)\psi_{t-h,T}(\widehat\theta)',\]
where the kernel $K$ and bandwidth $L_T$ satisfy \cref{asn:HAC-kernel}. Note that we do not center the sample autocovariances $\widehat\Gamma_{\psi,T}(h)$. We conjecture---but do not prove---that the centering is asymptotically negligible, since the proofs of \cref{thm:AN,lem:truncation-error} suggest that $\bar\psi_T(\widehat\theta)\equiv T^{-1}\sum_{t=1}^T\psi_{t,T}(\widehat\theta)=\widehat g(\widehat\theta)+ O_p(T^{-1})=O_p(T^{-1/2})$, so the difference between the centered and uncentered versions of $\widehat\Gamma_{\psi,T}(h)$ changes $\widehat\Omega$ by at most $O_p\left(L_T\|\bar\psi_T(\widehat\theta)\|^2\right)=O_p(L_T/T)=o_p(1)$ under \cref{asn:HAC-kernel:ii}.

\begin{thm}[Consistent HAC estimator: general case]
\label{thm:HAC-Omega}
Under \cref{asn:stationarity,asn:parameter-space,asn:weight-convergence,asn:clt} and \cref{asn:HAC-kernel,asn:stronger-stationarity,asn:Fourier-series-summability,asn:identification-general,asn:invertibility-general,asn:J-conjugate-symmetry}, 
\[
\widehat\Omega
\pto \Omega.
\]
\end{thm}

\begin{proof}
Define the following additional notation:
\begin{align*}
\Omega_T&\equiv \sum_{|h|<T}\Gamma_{\psi,T}(h) &\text{where}\quad \Gamma_{\psi,T}(h)\equiv\frac{1}{T}\sum_{t=1}^{T}\E\left[\psi_{t,T}(\theta_0)\psi_{t-h,T}(\theta_0)'\right], \\
\widetilde\Omega_T
&\equiv \sum_{|h|\le L_T} K(h/L_T)\widetilde\Gamma_{\psi,T}(h) &\text{where}\quad \widetilde\Gamma_{\psi,T}(h)\equiv \frac{1}{T}\sum_{t=1}^{T}\psi_{t,T}(\theta_0)\psi_{t-h,T}(\theta_0)'.
\end{align*}

By \cref{asn:stationarity:i,asn:J-conjugate-symmetry}, $\psi_t(\theta)$ and $\psi_{t,T}(\theta)$ are real-valued.
Note that \cref{asn:parameter-space,asn:Fourier-series-summability:i} imply \cref{asn:J-regularity}, and \cref{asn:Fourier-series-summability:i,asn:Fourier-series-summability:ii} imply \cref{asn:J-deriv},
so we have the full set of assumptions for consistency and asymptotic normality in \cref{thm:consistency,thm:AN}.

\medskip
\emph{(i) Plug-in: $\widehat\Omega-\widetilde\Omega_T\pto 0$.}
Let $\Delta_t \equiv \psi_{t,T}(\widehat\theta)-\psi_{t,T}(\theta_0)$. 
For each $h$, by Cauchy-Schwarz,
\begin{align*}
\|\widehat\Gamma_{\psi,T}(h)-\widetilde\Gamma_{\psi,T}(h)\|&=\left\|\frac{1}{T}\sum_{t=h+1}^T
\left(\Delta_t\psi_{t-h,T}(\theta_0)'+\psi_{t,T}(\theta_0)\Delta_{t-h}'+\Delta_t\Delta_{t-h}'\right)\right\|\\
&\le 2\left(\frac{1}{T}\sum_{t=1}^T\|\Delta_t\|^2\right)^{1/2}
\left(\frac{1}{T}\sum_{t=1}^T\|\psi_{t,T}(\theta_0)\|^2\right)^{1/2}
+ \frac{1}{T}\sum_{t=1}^T\|\Delta_t\|^2.
\end{align*}
Since the kernel $K(\cdot)$ is bounded and supported on $[-1,1]$ by \cref{asn:HAC-kernel:i},
only $|h|\le L_T$ contribute, so for some constant $C<\infty$,
\begin{align*}
\|\widehat\Omega-\widetilde\Omega_T\|
&\le C\sum_{|h|\le L_T}\|\widehat\Gamma_{\psi,T}(h)-\widetilde\Gamma_{\psi,T}(h)\|\\
&\le C L_T\left[
\left(\frac{1}{T}\sum_{t=1}^T\|\Delta_t\|^2\right)^{1/2}
\left(\frac{1}{T}\sum_{t=1}^T\|\psi_{t,T}(\theta_0)\|^2\right)^{1/2}
+\frac{1}{T}\sum_{t=1}^T\|\Delta_t\|^2
\right].
\end{align*}
Under \cref{asn:stationarity:i} and \cref{asn:Fourier-series-summability:i} at $\theta_0$, $\E\|\psi_{t,T}(\theta_0)\|^2$ is bounded uniformly in $t,T$. Hence,
\[\frac{1}{T}\sum_{t=1}^T\|\psi_{t,T}(\theta_0)\|^2=O_p(1).\]
Moreover, by the mean value theorem,
\[\frac{1}{T}\sum_{t=1}^T\|\Delta_t\|^2\le \frac{1}{T}\sum_{t=1}^T\sup_{\theta\in\Theta_0}\|\partial\psi_{t,T}(\theta)/\partial\theta'\|^2\|\widehat\theta-\theta_0\|^2.\]
\cref{thm:AN} implies that $\|\widehat\theta-\theta_0\|=O_p(T^{-1/2})$. Under \cref{asn:stationarity:i} and \cref{asn:Fourier-series-summability:ii}, we have that $\E[\sup_{\theta\in\Theta_0}\|\partial\psi_{t,T}(\theta)/\partial\theta'\|^2]$ is bounded uniformly in $t,T$. Therefore,
\[
\|\widehat\Omega-\widetilde\Omega_T\|
=O_p\left(\frac{L_T}{\sqrt{T}}\right)+O_p\left(\frac{L_T}{T}\right)
=o_p(1),
\]
because $L_T/\sqrt{T}\to 0$ by \cref{asn:HAC-kernel:ii}.

\medskip
\emph{(ii) Bias: $\E[\widetilde\Omega_T]-\Omega_T\to 0$.}
Define the scalar weights
\[
a_T(h)\equiv K\left(\frac{h}{L_T}\right)\1\{|h|\le L_T\}.
\]
Then
\[
\E[\widetilde\Omega_T]-\Omega_T=\sum_h \left(a_T(h)-1\right)\Gamma_{\psi,T}(h).
\]
For each fixed $h$, as $L_T\to\infty$ by \cref{asn:HAC-kernel:ii}, we have $\1\{|h|\le L_T\}\to 1$. Also, $K(h/L_T)\to K(0)=1$ by \cref{asn:HAC-kernel:i}, so $a_T(h)\to 1$ pointwise in $h$. Moreover, $|a_T(h)-1|$ is uniformly bounded since $K$ is.

Decompose $\Gamma_{\psi,T}(h)=\check\Gamma_{\psi,T}(h)+\frac1T\sum_{t=h+1}^T\E[\eta_{t,T}]\E[\eta_{t-h,T}]'$, where $\check\Gamma_{\psi,T}(h)$ is the centered part of \cref{lem:Gamma-psiT-envelope}, $\eta_{t,T}$ is the edge-truncation error of \cref{lem:truncation-error}, and we used that $\E[\psi_{t,T}(\theta_0)]=\E[\eta_{t,T}]$ by \cref{asn:identification-general} (which gives $\E[\psi_t(\theta_0)]=g(\theta_0)=0$).
For the centered part, \cref{lem:Gamma-psiT-envelope} (under \cref{asn:stationarity,asn:Fourier-series-summability:i} at $\theta_0$) gives $\|\check\Gamma_{\psi,T}(h)\|\le C(h)$ with $C(h)$ absolutely summable, so by dominated convergence $\sum_h |a_T(h)-1|\,\|\check\Gamma_{\psi,T}(h)\|\to 0$.
For the mean part, \cref{lem:truncation-error} gives $\sum_{t=1}^T\|\E[\eta_{t,T}]\|\le \sum_{t=1}^T(\E\|\eta_{t,T}\|^2)^{1/2}=O(1)$, so
\[
\sum_h\left|a_T(h)-1\right|\left\|\frac1T\sum_{t=h+1}^T\E[\eta_{t,T}]\E[\eta_{t-h,T}]'\right\|\le \frac{\sup_u |K(u)|+1}{T}\left(\sum_{t=1}^T\|\E[\eta_{t,T}]\|\right)^2=O(1/T).
\]
Combining the two parts, $\|\E[\widetilde\Omega_T]-\Omega_T\|\to 0$ as $T\to\infty$.

\medskip
\emph{(iii) Stochastic term: $\widetilde\Omega_T-\E[\widetilde\Omega_T]\pto 0$.}
We show $\var(\ve(\widetilde\Omega_T))\to0$, which gives the claim by Chebyshev's inequality. Since the kernel is bounded by \cref{asn:HAC-kernel:i}, there is a constant $C<\infty$ such that
\[
\left\|\var(\ve(\widetilde\Omega_T))\right\|
\le \frac{C}{T^2}\sum_{|h|\le L_T}\sum_{|\tau|\le L_T}\sum_{t=1}^T\sum_{s=1}^T
\left\|\cov\!\left(\ve(\psi_{t,T}(\theta_0)\psi_{t-h,T}(\theta_0)'),\,\ve(\psi_{s,T}(\theta_0)\psi_{s-\tau,T}(\theta_0)')\right)\right\|.
\]
Under \cref{asn:stationarity:i,asn:stronger-stationarity,asn:Fourier-series-summability:i}, \cref{lem:score-4th-summable} implies that the inner sum over $s$ is bounded uniformly in $t,h,\tau,T$, so
\[
\left\|\var(\ve(\widetilde\Omega_T))\right\|
\le \frac{C}{T^2}(2L_T+1)^2\,T\,\Lambda=O(L_T^2/T)\to0
\]
by \cref{asn:HAC-kernel:ii}. Hence $\widetilde\Omega_T-\E[\widetilde\Omega_T]\pto0$.

\medskip
\emph{(iv) Edge truncation: $\Omega_T\to\Omega$.} Let $\eta_{t,T}=\psi_{t,T}(\theta_0)-\psi_t(\theta_0)$ be the edge-truncation error of \cref{lem:truncation-error}. By that lemma, $\sum_{t=1}^T(\E\|\eta_{t,T}\|^2)^{1/2}=O(1)$, so by Cauchy-Schwarz,
\[
\E\left\|\frac{1}{\sqrt{T}}\sum_{t=1}^T\eta_{t,T}\right\|^2\le \frac{1}{T}\left(\sum_{t=1}^T\left(\E\|\eta_{t,T}\|^2\right)^{1/2}\right)^2=O(1/T)\to 0.
\]
Since
\[\frac{1}{\sqrt{T}}\sum_t\psi_{t,T}(\theta_0)=\frac{1}{\sqrt{T}}\sum_t\psi_t(\theta_0)+\frac{1}{\sqrt{T}}\sum_t\eta_{t,T},\] 
the second moments of $\frac{1}{\sqrt{T}}\sum_t\psi_{t,T}(\theta_0)$ and $\frac{1}{\sqrt{T}}\sum_t\psi_t(\theta_0)$ have the same limit.
Moreover, by \cref{asn:stationarity,asn:identification-general,asn:Fourier-series-summability:i} at $\theta_0$, \cref{rem:summability-Gamma-psi-general} gives that $\sum_{h\in\mathbb Z}\|\Gamma_\psi(h)\|<\infty$, so by dominated convergence,
\[
\frac1T \E\left[\left(\sum_{t=1}^T \psi_t(\theta_0)\right)\left(\sum_{t=1}^T \psi_t(\theta_0)\right)'\right]
=\sum_{|h|<T}\left(1-\frac{|h|}{T}\right)\Gamma_\psi(h)
\to \sum_{h\in\mathbb Z}\Gamma_\psi(h)=\Omega.
\]
Therefore, $\Omega_T\to\Omega$.

\medskip
\emph{(v) Conclusion.} Write
\[
\widehat\Omega-\Omega=\left(\widehat\Omega-\widetilde\Omega_T\right)+\left(\widetilde\Omega_T-\E[\widetilde\Omega_T]\right)+\left(\E[\widetilde\Omega_T]-\Omega_T\right)+\left(\Omega_T-\Omega\right).
\]
By (i)--(iv), each term is $o_p(1)$ or $o(1)$, so $\widehat\Omega\pto\Omega$.
\end{proof}

% ============================================================
\subsection{Auxiliary results for general spectral GMM}\label{sec:aux-general-spectral-gmm}

\subsubsection{Spectral regularity}
\begin{remark}[Continuity and boundedness of spectrum]
\label{rem:S-regularity}
Under \cref{asn:stationarity}, we have $\sum_{h\in\mathbb Z}\|\Gamma_\zeta(h)\|<\infty$, where $\Gamma_\zeta(h)\equiv \cov(\zeta_h,\zeta_0)$. Hence
$S(\omega)\equiv 1/(2\pi)\sum_{h\in\mathbb Z}\Gamma_\zeta(h)e^{-\iota h\omega}$ converges absolutely and uniformly in $\omega\in[0,2\pi]$, so $S(\omega)$ is continuous, $\sup_{\omega\in[0,2\pi]}\|S(\omega)\|<\infty$, and $\int_0^{2\pi}\|S(\omega)\|\,d\omega<\infty$.
\end{remark}

\begin{remark}[Boundedness of spectral weight function]
\label{rem:J-boundedness}
\cref{asn:parameter-space,asn:J-regularity} yield that \[\sup_{\theta\in\Theta}\sup_{\omega\in[0,2\pi]}\|\widetilde J(\omega;\theta)\|<\infty.\]
To see this, define $B(\theta) \equiv \sup_{\omega\in[0,2\pi]}\|\widetilde J(\omega;\theta)\|$, which is finite by \cref{asn:J-regularity:i}.
For any $\widetilde\theta,\theta\in\Theta$,
\[
|B(\widetilde\theta)-B(\theta)|
\le \sup_{\omega\in[0,2\pi]}\|\widetilde J(\omega;\widetilde\theta)-\widetilde J(\omega;\theta)\|
\le \Delta_{\widetilde J}(\|\widetilde\theta-\theta\|),
\]
so \cref{asn:J-regularity:ii} implies $B$ is continuous on compact $\Theta$, hence $\sup_{\theta\in\Theta}B(\theta)<\infty$.
\end{remark}

\begin{remark}[Well-definedness and continuity of population moments]
\label{rem:g-continuity}
Under \cref{asn:parameter-space,asn:stationarity,asn:J-regularity},
$g(\theta)$ is well-defined and continuous on $\Theta$.

To see this, note that under \cref{asn:stationarity:ii}, \cref{rem:S-regularity} yields that $\ve(S(\omega))$ is integrable on $[0,2\pi]$.
Moreover, under \cref{asn:parameter-space,asn:J-regularity}, \cref{rem:J-boundedness} implies that
\[\int_0^{2\pi}\|\widetilde J(\omega;\theta)\ve(S(\omega))\|\,d\omega<\infty\] uniformly in $\theta$,
so $g(\theta)$ is well-defined.

Fix $\theta\in\Theta$ and let $\widetilde\theta\to\theta$.
Then, using \cref{asn:J-regularity:ii},
\begin{align*}
\|g(\widetilde\theta)-g(\theta)\|
&\le \int_0^{2\pi}\|\widetilde J(\omega;\widetilde\theta)-\widetilde J(\omega;\theta)\|\|\ve(S(\omega))\|\,d\omega\\
&\le \Delta_{\widetilde J}(\|\widetilde\theta-\theta\|)\int_0^{2\pi}\|\ve(S(\omega))\|\,d\omega \to 0.
\end{align*}
Hence, $g$ is continuous on $\Theta$.
\end{remark}

% ============================================================
\subsubsection{Tools for consistency}
\begin{lem}[Mean of periodogram]
\label{lem:periodogram-mean}
Under \cref{asn:stationarity}, 
\[
\sup_{\omega\in[0,2\pi]}\|\E\widehat S(\omega)-S(\omega)\|\to 0.
\]
\end{lem}
This is the matrix analogue of the scalar uniform convergence result under mean zero in \citet[Proposition 10.3.1]{Brockwell1991}.
\begin{proof}

By \cref{asn:stationarity:ii}, \cref{rem:S-regularity} gives that $S(\omega)$ is continuous and bounded.
For the periodogram, the Fej\'er representation yields
\[
\E\widehat S(\omega)=\int_0^{2\pi}\widetilde{F}_T(\omega-\lambda)S(\lambda)\,d\lambda,
\]
where $\widetilde F_T$ is the Fej\'er kernel; see for example \citet[Section 2.4]{stoica2005spectral}. Since $\widetilde F_T$ is an approximate identity and $S$ is continuous and periodic on the compact set $[0,2\pi]$, the claim follows.
\end{proof}

\begin{lem}[Covariance bound for periodogram: \citealp{Brillinger1981}, Theorem 7.2.2]
\label{lem:periodogram-cov}
Under \cref{asn:stationarity}, there exist a constant $C<\infty$ and a deterministic sequence $c_T \to 0$ such that, uniformly in $j,j'\in\{0,\dots,T-1\}$,
\[
\left\|\cov\left(\ve(\widehat S(\omega_j)),\ve(\widehat S(\omega_{j'}))\right)\right\|
\le
\begin{cases}
C, & j=j' \ \text{or}\ j'=T-j,\\
c_T, & \text{otherwise}.
\end{cases}
\]
\end{lem}

\begin{proof}
Apply the proof of Theorem 7.2.2 in \citet[p.\ 433]{Brillinger1981}, but appeal to his Theorem 4.3.1 instead of Theorem 4.3.2 (pp.\ 92--93) to bound remainders.
\end{proof}

\begin{lem}[Pointwise convergence for spectral mean]
\label{lem:pointwise-linear-mean}
Let $m:[0,2\pi]\to\mathbb C^{d\times d_\zeta^2}$ be bounded and piecewise continuous.
Under \cref{asn:stationarity}, the spectral mean
\[
M_T(m)\equiv\frac{2\pi}{T}\sum_{j=0}^{T-1} m(\omega_j)\ve(\widehat S(\omega_j))
\pto
\int_0^{2\pi} m(\omega)\ve(S(\omega))\,d\omega.
\]
\end{lem}
See \citet[Chapter 7.6]{Brillinger1981} for related results with slightly different regularity conditions on $m$.
\begin{proof}
 Write
\[
M_T(m)-\int m\ve(S)
=\underbrace{\left(M_T(m)-\E M_T(m)\right)}_{\text{stochastic term}}
+\underbrace{\left(\E M_T(m)-\int m\ve(S)\right)}_{\text{bias term}}.
\]

\medskip
\emph{(i) Bias term.}
By \cref{lem:periodogram-mean},
\[
\left\|\E M_T(m)-\frac{2\pi}{T}\sum_{j=0}^{T-1} m(\omega_j)\ve(S(\omega_j))\right\|
\le \sup_\omega\|m(\omega)\|\times \frac{2\pi}{T}\sum_{j=0}^{T-1}\|\E\widehat S(\omega_j)-S(\omega_j)\|\to 0.
\]
Since $\omega\mapsto m(\omega)\ve(S(\omega))$ is piecewise continuous and bounded, the Riemann sum converges:
\[
\frac{2\pi}{T}\sum_{j=0}^{T-1} m(\omega_j)\ve(S(\omega_j))\to \int_0^{2\pi} m(\omega)\ve(S(\omega))\,d\omega.
\]
Therefore, $\E M_T(m)-\int m\ve(S)\to 0$.

\medskip
\emph{(ii) Stochastic term.} Under \cref{asn:stationarity}, \cref{lem:periodogram-cov} bounds $\|\cov(\ve(\widehat S(\omega_j)),\ve(\widehat S(\omega_{j'})))\|$ uniformly in $j,j'$.
Using boundedness of $m$, we have
\begin{align*}
\E\left\|\ve\left(M_T(m)-\E M_T(m)\right)\right\|^2
&\le \left(\frac{2\pi}{T}\right)^2\sum_{j,j'}
\|m(\omega_j)\|\|m(\omega_{j'})\|
\left\|\cov\left(\ve(\widehat S(\omega_j)),\ve(\widehat S(\omega_{j'}))\right)\right\|\\
&=o(1).
\end{align*}
Thus, $M_T(m)-\E M_T(m)\pto 0$ by Chebyshev's inequality. 

\medskip
Finally, combining (i) and (ii) yields the claim.
\end{proof}

\begin{lem}[Stochastic equicontinuity of sample moments]
\label{lem:SE}
Under \cref{asn:stationarity,asn:J-regularity},
the process $\{\widehat g(\theta):\theta\in\Theta\}$ is stochastically equicontinuous, i.e.,
for every $\varepsilon>0$,
\[
\lim_{\delta\to 0}\limsup_{T\to\infty}\Pr\left(
\sup_{\|\widetilde\theta-\theta\|<\delta}\|\widehat g(\widetilde\theta)-\widehat g(\theta)\|>\varepsilon
\right)=0.
\]
\end{lem}
\begin{proof}
For any $\widetilde\theta,\theta\in\Theta$, by Cauchy-Schwarz,
\begin{align*}
\|\widehat g(\widetilde\theta)-\widehat g(\theta)\|
&\le \frac{2\pi}{T}\sum_{j=0}^{T-1}\|\widetilde J(\omega_j;\widetilde\theta)-\widetilde J(\omega_j;\theta)\|\|\ve(\widehat S(\omega_j))\|\\
&\le \left(\frac{2\pi}{T}\sum_{j=0}^{T-1}\|\widetilde J(\omega_j;\widetilde\theta)-\widetilde J(\omega_j;\theta)\|^2\right)^{1/2}\left(\frac{2\pi}{T}\sum_{j=0}^{T-1}\|\widehat S(\omega_j)\|^2\right)^{1/2}.
\end{align*}

It follows from \cref{lem:periodogram-mean,lem:periodogram-cov} that $\E\|\widehat S(\omega_j)\|^2$ is uniformly bounded in $j$. Hence, by Markov's inequality,
\[V_T \equiv \frac{2\pi}{T}\sum_{j=0}^{T-1}\|\widehat S(\omega_j)\|^2=O_p(1).\]

Under \cref{asn:J-regularity:ii}, $\|\widetilde J(\omega_j;\widetilde\theta)-\widetilde J(\omega_j;\theta)\|\le \Delta_{\widetilde J}(\delta)$ for $\|\widetilde\theta-\theta\|<\delta$ and all $j$, hence
\[
\sup_{\|\widetilde\theta-\theta\|<\delta}\|\widehat g(\widetilde\theta)-\widehat g(\theta)\|
\le \sqrt{2\pi}\Delta_{\widetilde J}(\delta)\sqrt{V_T}.
\]
Fix $\varepsilon>0$ and $M>0$. Then
\[
\Pr\left(\sup_{\|\widetilde\theta-\theta\|<\delta}\|\widehat g(\widetilde\theta)-\widehat g(\theta)\|>\varepsilon\right)
\le \Pr\left(\sqrt{V_T}>M\right) + \1\left(\sqrt{2\pi}\Delta_{\widetilde J}(\delta)M>\varepsilon\right).
\]
As $V_T=O_p(1)$, choose $M$ large to make the first term small uniformly in $T$. Then choose $\delta$ small so that $\sqrt{2\pi}\Delta_{\widetilde J}(\delta)M\le \varepsilon$ using \cref{asn:J-regularity:ii}.
This proves stochastic equicontinuity.
\end{proof}

The bound $V_T=O_p(1)$ is sufficient here because we only need to control a fixed-weight, finite-dimensional parametric spectral mean. In contrast, results that control an empirical spectral process uniformly over a rich index class, as in \citet{Dahlhaus1988}, typically require stronger smoothness conditions on the spectrum and an entropy restriction on the index class.

\begin{lem}[Uniform convergence of sample moments]
\label{lem:uniform-g}
Under \cref{asn:parameter-space,asn:stationarity,asn:J-regularity},
\[
\sup_{\theta\in\Theta}\|\widehat g(\theta)-g(\theta)\|\pto 0.
\]
\end{lem}
\begin{proof}
\cref{asn:parameter-space} gives compact $\Theta$.
Under \cref{asn:parameter-space,asn:stationarity,asn:J-regularity}, \cref{rem:g-continuity} gives continuity of $g$ on $\Theta$.
Under \cref{asn:parameter-space,asn:stationarity,asn:J-regularity}, \cref{lem:SE} gives stochastic equicontinuity.

Under \cref{asn:stationarity}, \cref{lem:pointwise-linear-mean} applies with $m(\omega)=\widetilde J(\omega;\theta)$ for each $\theta$, which is bounded and piecewise continuous by \cref{asn:J-regularity:i}, so $\widehat g(\theta)\pto g(\theta)$ pointwise.

Together, the uniform convergence follows from pointwise convergence plus stochastic equicontinuity on compact 
$\Theta$ and continuity of $g$ \citep[Lemma 2.8]{NeweyMcFadden1994}.
\end{proof}

% ============================================================
\subsubsection{Tools for asymptotic normality}
\begin{remark}[Continuity of $G$]
\label{rem:G-continuity}
Under \cref{asn:stationarity,asn:J-deriv}, $G(\theta)$ is continuous on $\Theta_0$. The argument is similar to \cref{rem:g-continuity}.
\end{remark}

\begin{lem}[Uniform convergence of $\widehat G$]
\label{lem:uniform-G}
Under \cref{asn:stationarity,asn:J-deriv},
\[
\sup_{\theta\in\Theta_0}\|\widehat G(\theta)-G(\theta)\|\pto 0.
\]
\end{lem}
The proof is similar to that of \cref{lem:uniform-g}.

% ============================================================

\subsubsection{Tools for HAC and time-domain CLT}
\begin{remark}[\cref{asn:Fourier-series-summability} vs.\ \cref{asn:J-regularity,asn:J-deriv}]
\label{rem:Fourier-summability-bullets}
Note that:
\begin{itemize}[leftmargin=2em]
\item \cref{asn:Fourier-series-summability:i} $\Rightarrow$ \cref{asn:J-regularity:i}: by the Weierstrass M-test, $\widetilde J(\cdot;\theta)$ is bounded and continuous in $\omega$ on $[0,2\pi]$, hence piecewise continuous and Riemann integrable.
\item \cref{asn:parameter-space,asn:Fourier-series-summability:i} $\Rightarrow$ \cref{asn:J-regularity:ii}: uniform absolute summability plus coefficient continuity on compact set implies the Fourier series is uniformly continuous by the Weierstrass M-test.
\item \cref{asn:Fourier-series-summability:i,asn:Fourier-series-summability:ii} $\Rightarrow$ \cref{asn:J-deriv:i,asn:J-deriv:ii,asn:J-deriv:iii}: similar argument on $\partial\ve(\widetilde J(\cdot;\theta))/\partial\theta'$ for $\theta\in\Theta_0$.
\item \cref{asn:J-regularity,asn:J-deriv} do not generally imply \cref{asn:Fourier-series-summability}:
even continuous differentiability in $\omega$ does not guarantee absolute summability of Fourier coefficients. See \cref{rem:Fourier-series-summability-sufficient} for a set of sufficient conditions.
\end{itemize}
\end{remark}

\begin{remark}[Sufficient conditions for \cref{asn:Fourier-series-summability}]
\label{rem:Fourier-series-summability-sufficient}
Assume that the map $\omega \mapsto \widetilde J(\omega;\theta)$ is twice differentiable with second derivative continuous in $(\omega,\theta)$, and that the following boundary conditions hold: $\widetilde J(0;\theta)=\widetilde J(2\pi;\theta)$ and $\partial \widetilde J(0;\theta)/\partial \omega =\partial \widetilde J(2\pi;\theta)/\partial \omega$. Then integration by parts gives
\[\widetilde{J}_k(\theta) = -\frac{1}{k^2} \times \frac{1}{2\pi}\int_0^{2\pi} e^{-\iota\omega k} \frac{\partial^2 \widetilde J(\omega;\theta)}{\partial \omega^2}\,d\omega,\]
which is easily seen to imply \cref{asn:Fourier-series-summability:i}. Similar conditions placed on the map $\omega \mapsto \partial \ve(\widetilde{J}(\omega;\theta))/\partial \theta'$ yield \cref{asn:Fourier-series-summability:ii}. Finally, under the conditions above, the display above implies (for any indices $j_1,j_2$)
\begin{align*}
\sum_{k=1}^\infty k |(\widetilde{J}_k(\theta_0))_{j_1,j_2}| &= \sum_{k=1}^\infty \frac{1}{k} \times \left|\frac{1}{2\pi}\int_0^{2\pi} e^{-\iota\omega k} \frac{\partial^2 (\widetilde J(\omega;\theta_0))_{j_1,j_2}}{\partial \omega^2}\,d\omega\right| \\
&\leq \left(\sum_{k=1}^\infty k^{-2}\right)^{1/2} \times \left(\sum_{k=1}^\infty \left|\frac{1}{2\pi}\int_0^{2\pi} e^{-\iota\omega k} \frac{\partial^2 (\widetilde J(\omega;\theta_0))_{j_1,j_2}}{\partial \omega^2}\,d\omega\right|^2\right)^{1/2} \\
&\leq \left(\sum_{k=1}^\infty k^{-2}\right)^{1/2} \times \left(\sum_{k \in \mathbb{Z}} \left|\frac{1}{2\pi}\int_0^{2\pi} e^{-\iota\omega k} \frac{\partial^2 (\widetilde J(\omega;\theta_0))_{j_1,j_2}}{\partial \omega^2}\,d\omega\right|^2\right)^{1/2} \\
&= \left(\sum_{k=1}^\infty k^{-2}\right)^{1/2} \times \left(\frac{1}{2\pi}\int_0^{2\pi} \left|\frac{\partial^2 (\widetilde J(\omega;\theta_0))_{j_1,j_2}}{\partial \omega^2}\right|^2\,d\omega\right)^{1/2},
\end{align*}
where the first inequality uses Cauchy-Schwarz and the final equality uses Parseval. Since $\partial^2 \widetilde{J}(\omega;\theta_0)/\partial \omega^2$ was already assumed to be bounded in $\omega$, \cref{asn:Fourier-series-summability:iii} follows.
\end{remark}

\begin{remark}[Summability of $\Gamma_\psi(h)$]
\label{rem:summability-Gamma-psi-general}
Under \cref{asn:stationarity,asn:identification-general,asn:Fourier-series-summability:i} at $\theta_0$, we have that $\sum_{h\in\mathbb Z}\|\Gamma_\psi(h)\|<\infty$, so $\Omega\equiv \sum_{h\in\mathbb Z}\Gamma_\psi(h)$ is well-defined. Indeed, by \cref{asn:identification-general}, $\E[\psi_t(\theta_0)]=g(\theta_0)=0$, so $\Gamma_\psi(h)=\E[\psi_h(\theta_0)\psi_0(\theta_0)']=\cov(\psi_h(\theta_0),\psi_0(\theta_0))$, and its absolute summability follows from the argument in \cref{lem:Gamma-psiT-envelope} below applied to the untruncated score $\psi_t(\theta_0)$.
\end{remark}

\begin{lem}[Truncation error]
\label{lem:truncation-error}
Let $\mathcal{K}_{t,T}^c\equiv\{k\in\mathbb Z:k<-(t-1)\ \text{or}\ k>T-t\}$ and define the edge-truncation error
\[
\eta_{t,T}\equiv \psi_{t,T}(\theta_0)-\psi_t(\theta_0)=-\sum_{k\in \mathcal{K}_{t,T}^c}\widetilde J_k(\theta_0)\ve(\zeta_{t+k}\zeta_t').
\]
Under \cref{asn:stationarity:i} and \cref{asn:Fourier-series-summability:iii} at $\theta_0$,
\[
\sum_{t=1}^T\left(\E\|\eta_{t,T}\|^2\right)^{1/2}=O(1).
\]
\end{lem}

\begin{proof}
Under \cref{asn:stationarity:i}, $\E\|\eta_{t,T}\|^2\le C\left(\sum_{k\in \mathcal{K}_{t,T}^c}\|\widetilde J_k(\theta_0)\|\right)^2$ for some $C<\infty$, so
\[
\sum_{t=1}^T\left(\E\|\eta_{t,T}\|^2\right)^{1/2}\le C^{1/2}\sum_{t=1}^T\sum_{k\in \mathcal{K}_{t,T}^c}\|\widetilde J_k(\theta_0)\|.
\]
Under \cref{asn:Fourier-series-summability:iii},
\begin{align*}
\sum_{t=1}^T\sum_{k\in \mathcal{K}_{t,T}^c}\|\widetilde J_k(\theta_0)\|&=\sum_{|k|\le T-1}|k|\|\widetilde J_k(\theta_0)\|+T\sum_{|k|\ge T}\|\widetilde J_k(\theta_0)\|\\
&\le \sum_{|k|\le T-1}|k|\|\widetilde J_k(\theta_0)\|+\sum_{|k|\ge T}|k|\|\widetilde J_k(\theta_0)\|=O(1),
\end{align*}
which gives the result.
\end{proof}

\begin{lem}[Summability of the centered score autocovariance]
\label{lem:Gamma-psiT-envelope}
Define
\[\check\Gamma_{\psi,T}(h)\equiv \frac1T\sum_{t=1}^T\cov\left(\psi_{t,T}(\theta_0),\psi_{t-h,T}(\theta_0)\right).\]
Under \cref{asn:stationarity} and \cref{asn:Fourier-series-summability:i} at $\theta_0$,
there exists a positive, deterministic sequence $C(h)$, independent of $T$, such that
$\|\check\Gamma_{\psi,T}(h)\|\le C(h)$ for all $T,h$, and $\sum_{h \in \mathbb{Z}} C(h)<\infty$.
\end{lem}

\begin{proof}
With $\mathcal{K}_{t,T}\equiv\{-(t-1),\dots,T-t\}$, $\psi_{t,T}(\theta_0)=\sum_{k\in \mathcal{K}_{t,T}}\widetilde J_k(\theta_0)\ve(\zeta_{t+k}\zeta_t')$.
Then
\[\check\Gamma_{\psi,T}(h) =\frac1T\sum_{t=1}^T\sum_{k\in \mathcal{K}_{t,T}}\sum_{\ell\in \mathcal{K}_{t-h,T}}
\widetilde J_k(\theta_0)\cov\left(\ve(\zeta_{t+k}\zeta_t'),\ve(\zeta_{t-h+\ell}\zeta_{t-h}')\right)\widetilde J_\ell(\theta_0)'.\]
Hence, by the triangle inequality, stationarity in \cref{asn:stationarity:i}, and dropping the finite index
restrictions,
\[
\|\check\Gamma_{\psi,T}(h)\|
\le
\sum_{k\in\mathbb Z}\sum_{\ell\in\mathbb Z}\|\widetilde J_k(\theta_0)\|\|\widetilde J_\ell(\theta_0)\|\left\|\cov\left(\ve(\zeta_k\zeta_0'),\ve(\zeta_{-h+\ell}\zeta_{-h}')\right)\right\|
\equiv C(h),
\]
which is independent of $T$.

To show summability of $C(h)$, note that each entry of $\cov(\ve(\zeta_k\zeta_0'),\ve(\zeta_{-h+\ell}\zeta_{-h}'))$ is a fourth-order covariance $\cov(AB,CD)$ with $A,B,C,D$ mean-zero components of $\zeta$, and admits the product cumulant decomposition \citep[Theorem 2.3.2]{Brillinger1981}
\[\cov(AB,CD)=\cum(A,B,C,D)+\cov(A,C)\cov(B,D)+\cov(A,D)\cov(B,C).\]
Thus, by \cref{asn:stationarity:ii},
\[
\sum_{h\in\mathbb Z}\left\|\cov\left(\ve(\zeta_k\zeta_0'),\ve(\zeta_{-h+\ell}\zeta_{-h}')\right)\right\|
\le M<\infty,
\]
uniformly in $k,\ell$.
Therefore, by \cref{asn:Fourier-series-summability:i},
\[
\sum_{h\in\mathbb Z} C(h)
\le
M\sum_{k,\ell\in\mathbb Z}\|\widetilde J_k(\theta_0)\|\|\widetilde J_\ell(\theta_0)\|
= M\left(\sum_{k\in\mathbb Z}\|\widetilde J_k(\theta_0)\|\right)^2
<\infty. \qedhere
\]
\end{proof}

\begin{lem}[Fourth-order dependence of the score]
\label{lem:score-4th-summable}
Under \cref{asn:stationarity:i,asn:stronger-stationarity} and \cref{asn:Fourier-series-summability:i} at $\theta_0$, there is a constant $\Lambda<\infty$ such that, for all $h,\tau\in\mathbb Z$ and all $t,T$,
\[
\sum_{s\in\mathbb Z}
\left\|\cov\!\left(\ve(\psi_{t,T}(\theta_0)\psi_{t-h,T}(\theta_0)'),\,\ve(\psi_{s,T}(\theta_0)\psi_{s-\tau,T}(\theta_0)')\right)\right\|
\le \Lambda.
\]
\end{lem}

\begin{proof}
To economize on notation, we drop the $\theta_0$ argument and pretend that $\widetilde J$ and $\zeta_t$ are scalars (this is without essential loss of generality, since all variables are finite-dimensional anyway). As in the proof of \cref{lem:Gamma-psiT-envelope}, we extend all lag sums to $k\in\mathbb Z$ after taking absolute values, since dropping the edge restrictions $\mathcal K_{t,T}$ only enlarges the bound, and write $\psi_{t,T}=\sum_{k\in\mathbb Z}\widetilde J_k\zeta_{t+k}\zeta_t$. Then
\[
|\cov\left(\psi_{t,T}\psi_{t-h,T},\psi_{s,T}\psi_{s-\tau,T}\right)|
\leq \sum_{k_1,k_2,\ell_1,\ell_2\in\mathbb Z}
|\widetilde J_{k_1}\widetilde J_{k_2}\widetilde J_{\ell_1}\widetilde J_{\ell_2}|\,
|\cov(P_{t,h,k_1,k_2},Q_{s,\tau,\ell_1,\ell_2})|,
\]
where $P_{t,h,k_1,k_2}\equiv \zeta_{t+k_1}\zeta_t\zeta_{t-h+k_2}\zeta_{t-h}$ and
$Q_{s,\tau,\ell_1,\ell_2}\equiv\zeta_{s+\ell_1}\zeta_s\zeta_{s-\tau+\ell_2}\zeta_{s-\tau}$. The lemma follows from \cref{asn:Fourier-series-summability:i} at $\theta_0$ if we can show that
\begin{equation} \label{eqn:summable-4th-order}
\sum_{s \in \mathbb{Z}} |\cov(P_{t,h,k_1,k_2},Q_{s,\tau,\ell_1,\ell_2})| \leq \Lambda
\end{equation}
uniformly in $h,\tau,k_1,k_2,\ell_1,\ell_2$.

By the product theorem for cumulants \citep[Theorem 2.3.2]{Brillinger1981},
\[\cov(P_{t,h,k_1,k_2},Q_{s,\tau,\ell_1,\ell_2}) = \sum_{\nu} \cum(\zeta_j; j \in \nu_1) \times \cdots \times \cum(\zeta_j; j \in \nu_p),\]
where the sum is over all those (finitely many) partitions $\nu=\nu_1 \cup \dots \cup \nu_p$ of the index list $\lbrace t+k_1,t,t-h+k_2,t-h,s+\ell_1,s,s-\tau+\ell_2,s-\tau\rbrace$ that are ``indecomposable'' in the language of \citeauthor{Brillinger1981}: at least one of the component sets---let it be $\nu_1$ upon reordering---must contain at least one index from the $t$ group and at least one from the $s$ group of indices. By \cref{asn:stronger-stationarity}, there exists $C\in [1,\infty)$ (depending only on $\E[\zeta_0^8]$) such that $|\cum(\zeta_j; j \in \nu_b)| \leq C$ for all component sets $\nu_b$. Thus,
\begin{equation} \label{eqn:cumulant_decomposition_bound}
|\cov(P_{t,h,k_1,k_2},Q_{s,\tau,\ell_1,\ell_2})| \leq C^6 \sum_{\nu} |\cum(\zeta_j; j \in \nu_1)|.
\end{equation}
By stationarity (\cref{asn:stationarity:i}) and the above-mentioned structure of $\nu_1$, each $\cum(\zeta_j; j \in \nu_1)$ can be written as
\[\cum(\zeta_0,\zeta_{\delta_1},\dots,\zeta_{\delta_{q_1}},\zeta_{s-t+\delta_{q_1+1}},\dots,\zeta_{s-t+\delta_{q_2}})\]
for some integers $\delta_1,\dots,\delta_{q_1},\delta_{q_1+1},\dots,\delta_{q_2}$, where $0 \leq q_1 < q_2 \leq 7$. The above cumulant is absolutely summable in $s$:
\[\sum_{s \in \mathbb{Z}}|\cum(\zeta_0,\zeta_{\delta_1},\dots,\zeta_{\delta_{q_1}},\zeta_{s-t+\delta_{q_1+1}},\dots,\zeta_{s-t+\delta_{q_2}})| \leq \sum_{\widetilde{\ell}_1,\dots,\widetilde{\ell}_{q_2} \in \mathbb{Z}}|\cum(\zeta_0,\zeta_{\widetilde\ell_1},\dots,\zeta_{\widetilde\ell_{q_2}})|,\]
and the latter sum is finite by \cref{asn:stronger-stationarity}; it is also independent of $h,\tau,k_1,k_2,\ell_1,\ell_2,t,T$. It follows from this result and from \eqref{eqn:cumulant_decomposition_bound} that \eqref{eqn:summable-4th-order} holds. This concludes the proof.
\end{proof}

\begin{remark}[Time-domain CLT]
\label{rem:time-domain-CLT}
Standard Fourier transform calculations yield
\[\int_0^{2\pi}\widetilde{J}(\omega;\theta_0)\ve(\widehat{S}(\omega))\,d\omega =\frac{1}{T}\sum_{t=1}^T \psi_{t,T}(\theta_0).\]
Since $\widehat g(\theta_0)= \frac{2\pi}{T}\sum_{j=0}^{T-1}\widetilde J(\omega_j;\theta_0)\ve(\widehat S(\omega_j))$ is a Riemann approximation to the integral on the left-hand side above, it can be shown that they differ by $O_p(T^{-1})$ under \cref{asn:stationarity,asn:Fourier-series-summability}. Moreover, by \cref{asn:stationarity,asn:Fourier-series-summability:iii}, step (iv) in the proof of \cref{thm:HAC-Omega} gives that
$\frac{1}{T}\sum_{t=1}^T \left(\psi_t(\theta_0)-\psi_{t,T}(\theta_0)\right)=o_p(T^{-1/2})$. Putting these together,
\[\widehat g(\theta_0)=\frac{1}{T}\sum_{t=1}^T \psi_t(\theta_0)+o_p(T^{-1/2}).\]
Consequently, the spectral mean CLT in \cref{asn:clt} follows from any conventional time-domain CLT for the stationary process $\{\psi_t(\theta_0)\}$.
\end{remark}

% ============================================================
\subsection{Specialization to main application}\label{sec:specialization-our-setup}
Now we specialize the general spectral GMM results to the application in \cref{sec:framework,sec:theory}. Recall that this framework corresponds to the spectral weight function $\widetilde J_\text{special}(\omega;\theta) \equiv M_z\otimes \left(I_{d_y},\,-J(\omega;\theta),\,0_{d_y\times d_z}\right)$, as discussed at the beginning of \cref{app:proofs}.

\subsubsection{Consistency}
\begin{remark}[Conjugate symmetry]
\label{rem:jacobian-real-implies-J-conjugate-symmetry}
Since $J_k(\theta)$ is real, \cref{asn:J-conjugate-symmetry} is satisfied for $\widetilde J_\text{special}(\omega;\theta)$.
\end{remark}

\begin{lem}[SSJ assumptions imply $\widetilde J$ regularity]
\label{lem:lowlevel-to-tildeJ}
Under \cref{asn:parameter-space,asn:jacobian}, $\widetilde J_\text{special}(\omega;\theta)$ satisfies \cref{asn:J-regularity}.
\end{lem}

\begin{proof}
Let $\Xi\equiv [0,2\pi]\times\Theta$. By \cref{asn:parameter-space}, $\Xi$ is compact.
Let $m_k(\omega;\theta)=e^{\iota\omega k}J_k(\theta)$ for each $k\ge 0$. By \cref{asn:jacobian:i}, $m_k$ is continuous on $\Xi$, and by \cref{asn:jacobian:ii},
\[
\sum_{k=0}^\infty \sup_{(\omega,\theta)\in\Xi}\left\|m_k(\omega;\theta)\right\|
=\sum_{k=0}^\infty \sup_{\theta\in\Theta}\left\|J_k(\theta)\right\|
<\infty.
\]
Hence, by the Weierstrass M-test, $J(\omega;\theta)\equiv \sum_{k=0}^\infty m_k(\omega;\theta)=\sum_{k=0}^\infty e^{\iota\omega k}J_k(\theta)$ converges uniformly on $\Xi$. Since each $m_k$ is continuous, $J$ is continuous on $\Xi$, and by Heine-Cantor it is uniformly continuous on $\Xi$.

For each $\theta\in\Theta$, continuity of $\omega\mapsto J(\omega;\theta)$ on $[0,2\pi]$ implies that $J(\cdot;\theta)$ is bounded and Riemann integrable in $\omega$.
Moreover, the uniform continuity of $J$ on $\Xi$ implies
\[
\sup_{\omega\in[0,2\pi]}\left\|J(\omega;\widetilde\theta)-J(\omega;\theta)\right\|\to 0,
\quad\text{as }\left\|\widetilde\theta-\theta\right\|\to 0,
\]
so $J(\cdot;\theta)$ is uniformly continuous in $\theta$ uniformly over $\omega$.

Since $\widetilde{J}_\text{special}$ is an affine transformation of $J$, the claim follows.
\end{proof}

\begin{proof}[Proof of \cref{prop:consistency}]
Note first that
\[g(\theta) = \int_0^{2\pi} \widetilde{J}_\text{special}(\omega;\theta)\ve(S(\omega))\,d\omega = \ve\left(\cov(y_t,z_t)-\sum_{k=0}^\infty J_k(\theta)\cov(x_{t+k},z_t)\right),\]
after interchanging the $k$-sum with the $\omega$-integral (justified by \cref{asn:stationarity,asn:jacobian}). Hence, \cref{asn:moment_cond,asn:identification} imply \cref{asn:identification-general} with $\widetilde{J}=\widetilde{J}_\text{special}$. \cref{rem:jacobian-real-implies-J-conjugate-symmetry} and \cref{lem:lowlevel-to-tildeJ} therefore imply that all the assumptions of \cref{thm:consistency} hold.
\end{proof}

\subsubsection{Asymptotic normality}
\begin{remark}[SSJ assumptions imply $\partial \widetilde J$ regularity]
\label{rem:lowlevel-to-tildeJ-deriv}
Under \cref{asn:parameter-space,asn:jacobian} and interiority of $\theta_0$ in $\Theta$, $\widetilde J_\text{special}$ satisfies \cref{asn:J-deriv} with $\Theta_0$ given by any compact neighborhood of $\theta_0$ contained in $\Theta$. The argument uses the same M-test and Heine-Cantor steps as in \cref{lem:lowlevel-to-tildeJ}.
\end{remark}

\begin{proof}[Proof of \cref{prop:AN}]
As in the proof of \cref{prop:consistency}, \cref{asn:moment_cond,asn:identification,asn:stationarity,asn:jacobian} imply \cref{asn:identification-general}, and a similar argument shows that \cref{asn:invertibility,asn:stationarity,asn:jacobian} imply \cref{asn:invertibility-general}, both with $\widetilde{J}=\widetilde{J}_\text{special}$. Hence, \cref{rem:jacobian-real-implies-J-conjugate-symmetry,rem:lowlevel-to-tildeJ-deriv} and \cref{lem:lowlevel-to-tildeJ} imply that all the assumptions of \cref{thm:AN} hold.
\end{proof}

\subsubsection{HAC}
\begin{proof}[Proof of \cref{prop:HAC-Omega}]
Note first that the general HAC estimator $\widehat\Omega$ from \cref{sec:est-var} reduces to the special HAC estimator in \cref{app:hac} when $\widetilde{J}=\widetilde{J}_\text{special}$, since
\[\widetilde{J}_{k,\text{special}}(\theta) \equiv \frac{1}{2\pi}\int_0^{2\pi} e^{-\iota\omega k} \widetilde{J}_\text{special}(\omega;\theta)\,d\omega = M_z \otimes \Big(\1\{k=0\} \times I,-\1\{k\geq 0\} \times J_k(\theta),0\Big).\]
Next, \cref{rem:jacobian-real-implies-J-conjugate-symmetry} gives \cref{asn:J-conjugate-symmetry}.
Moreover, with $\widetilde{J}=\widetilde{J}_\text{special}$, \cref{asn:moment_cond,asn:identification,asn:stationarity,asn:jacobian} imply \cref{asn:identification-general}, \cref{asn:invertibility,asn:stationarity,asn:jacobian} imply \cref{asn:invertibility-general}, \cref{asn:jacobian} implies \cref{asn:Fourier-series-summability:i,asn:Fourier-series-summability:ii}, and \cref{asn:jacobian-tail-summability} implies \cref{asn:Fourier-series-summability:iii}. Thus, the conclusion follows from \cref{thm:HAC-Omega}.
\end{proof}

\clearpage
\phantomsection
\addcontentsline{toc}{section}{References}
\bibliography{ref}

\end{appendices}

\end{document}